%% file: TowardsThQuantumComp.tex
\documentclass[11pt]{amsart}

\usepackage{hyperref}
\usepackage{xspace}

\usepackage[utf8]{inputenc}

\usepackage{amssymb,amsmath,amsxtra,stmaryrd,verbatim,framed,wasysym}

\usepackage{color}
\usepackage[usenames,dvipsnames,svgnames,table]{xcolor}
\usepackage{ulem}
\usepackage{subfigure}
\usepackage{graphics}

\newtheorem{theorem}{Theorem}
\newtheorem{corollary}[theorem]{Corollary}
\newtheorem{proposition}[theorem]{Proposition}
\newtheorem{lemma}[theorem]{Lemma}

\theoremstyle{definition}
\newtheorem{definition}[theorem]{Definition}

\theoremstyle{remark}
\newtheorem{remark}[theorem]{Remark}
\newtheorem{example}[theorem]{Example}

\newcommand{\varletter}[1]{\mathcal{#1}}
\newcommand{\Q}{\varletter{Q}}
\newcommand{\varS}{\varletter{S}}
\newcommand{\T}{\varletter{T}}
\newcommand{\F}{\varletter{F}}
\newcommand{\K}{\varletter{K}}
\newcommand{\B}{\varletter{B}}

\newcommand{\V}{\varletter{V}}
\newcommand{\varP}{\varletter{P}} 

\newcommand{\dletter}[1]{\mathbb{#1}}
\newcommand{\DD}{\dletter{D}} 
\newcommand{\ZZ}{\dletter{Z}}
\newcommand{\CC}{\dletter{C}}  
\newcommand{\NN}{\dletter{N}} 
\newcommand{\PP}{\dletter{P}} 

\newcommand{\norm}[1]{\left\Vert#1\right\Vert}
\newcommand{\inprod}[2]{\left\langle #1, #2 \right\rangle}
\newcommand{\kinprod}[2]{\left\langle \mathsf{#1}\mid \mathsf{#2} \right\rangle}

\newcommand{\ket}[1]{\left\vert\mathsf{#1}\right\rangle} 
\newcommand{\tek}[1]{\left\langle\mathsf{#1}\right\vert} 
\newcommand{\cb}[1]{\mathsf{CB}(#1)}
\newcommand{\SPAN}[1]{\mathsf{span}(#1)}
\newcommand{\CoC}{\widetilde{\CC}}

\newcommand{\partialto}{\rightharpoonup}
\newcommand{\CConf}{\mathfrak{C}}
\newcommand{\GConf}{\mathfrak{C}}
\newcommand{\FConf}{\mathfrak{F}}
\newcommand{\SConf}[1]{\GConf^{#1}}
\newcommand{\TConf}[1]{\overline{\GConf}^{#1}}
\newcommand{\qConf}{\mathfrak{q}\GConf}
\newcommand{\qFConf}{\mathfrak{q}\FConf}
\newcommand{\cqConf}{q\mathcal{C}}

\newcommand{\nstring}[1]{\underline{#1}}

\newcommand{\rdelta}{\overline\delta}
\newcommand{\rgamma}{\overline\gamma}

\newcommand{\lub}{\bigsqcup}

\newcommand{\BeV}{B\&\!V\xspace}

\newcommand{\val}{\mathbf{val}}

\newcommand{\qcf}{\mathcal{QCF}}

\newcommand{\pd}[1]{\mathbf{P}_{#1}}
\newcommand{\prob}[2]{\pd{#1}(#2)}

\newcommand{\outobs}[3]{\ket{#1}\downarrow_{#2}\ket{#3}}
\newcommand{\obsout}[2]{#1\downarrow_{#2}}
\newcommand{\PR}{\mathsf{Pr}}
\newcommand{\pr}[1]{\PR\{#1\}}

\newcommand{\abextra}[1]{\overline{#1}}
\newcommand{\aball}[1]{\widehat{#1}}
\newcommand{\symextra}[1]{\overline{#1}}

\newcommand{\dind}[2]{\genfrac{}{}{0pt}{1}{#1}{#2}}

\title{Towards A Theory Of Quantum Computability} 

\date{April 10, 2015}

\author[S. Guerrini]{Stefano Guerrini$^*$} \address{Stefano Guerrini\\
  LIPN, UMR 7030 CNRS, Institut Galilée, Université Paris13, Sorbonne
  Paris Cité} 

\email{stefano.guerrini@univ-paris13.fr}

\thanks{$^*$
  Partially supported by the Project ELICA (ref.~ANR-14-CE25-0005), of
  the ANR program ``Fondements du numérique (DS0705) 2014''.}

\author[S. Martini]{Simone Martini$^{\S}$} \address{Simone Martini\\ Dipartimento di
  Informatica -- Scienza e Ingegneria, Università di Bologna, and Inria Sophia-Antipolis}

\email{simone.martini@unibo.it} \thanks{$^{\S}$ Partially supported by the Italian ``National Group for Algebraic and Geometric Structures, and their Applications'' (GNSAGA-INDAM)}

\author[A. Masini]{Andrea Masini$^{\P}$}

\address{Andrea Masini\\ Dipartimento di Informatica, Università di Verona}

\email{andrea.masini@univr.it}

\thanks{$^{\P}$ Partially written while at LIPN, Institut Galilée,
  Université Paris13, Sorbonne Paris Cité as visiting researcher.}

\begin{document}

\begin{abstract} 
  We propose a definition of quantum computable functions as mappings
  between superpositions of natural numbers to probability
  distributions of natural numbers.  Each function is obtained as a
  limit of an infinite computation of a quantum Turing machine.  The
  class of quantum computable functions is recursively enumerable,
  thus opening the door to a quantum computability theory which may
  follow some of the classical developments.
\end{abstract}

\maketitle

\section{Introduction}

Despite the availability of a large corpus of results\footnote{
See, in the large literature, \cite{Deu85,MiOh05,NiOz02,NishOza10,Yao93} for fundamentals
results, \cite{BerVa97} for the foundations of quantum complexity, or 
\cite{AltGra05,DLZorzi15,DMZ10,DLMZentcs11,DLMZmscs08,DKP07,MVZ-jmvl,Sel04c,SelVal06,Zorzi15} for
more language oriented papers.
}, quantum
computability still lacks a general treatment akin to classical
computability theory. Taking as a reference model (quantum) Turing
machines, one of the main obstacles is that while it is obvious how to
understand a classical Turing machine (TM) as a device computing a
numerical function, the same is not so for a quantum Turing machine
(QTM).

In a na\"\i{}ve, but suggestive way, a QTM may be described as a
classical TM which, at any point of its computation, evolves into
several different classical configurations, each characterised by a
complex \emph{amplitude}.  Such different configurations should be
imagined as simultaneously present, ``in \emph{superposition}''---a
simultaneity formally expressed as a weighted sum
$\sum d_i\mathsf{C}_i$ of classical configurations $C_i$, with complex
coefficients $d_i$.  Even when starting from a classical
configuration, in a QTM, there is not a single result, but a
superposition from which we can read off several numerical ``results''
with certain probabilities\footnote{ We cannot observe the entirety of
  a superposition without destroying it. But if we insist on observing
  it, then we will obtain one of the classical configurations
  $\mathsf{C_i}$, with probability $|d_i|^2$.}.  Moreover, QTMs never
have genuinely finite computations and one should therefore find a way
to define when and how a result can be read off.

In this paper we propose a notion of ``function computable by a
quantum Turing machine,'' as a mapping between superpositions of
initial classical configurations\footnote{More precisely, the domain
  of our functions will be the Hilbert space $\ell_1^2(\NN)$ of square
  summable, denumerable sequences of complex numbers, with unitary norm.} to
probability distributions of natural numbers, which are obtained (in
general) as a limit of an infinite QTM computation.

Before reaching this point, however, we must go back to the basics,
and look to the very notion of QTM. Because, if it is true that
configurations of a QTM are superpositions $\sum d_i\mathsf{C}_i$ of
classical configurations, quantum physics principles impose severe
constraints on the possible evolutions of such machines. First, in any
superposition $\sum d_i\mathsf{C}_i$, we must have $\sum |d_i|^2 =1$.
Second, there cannot be any finite computation---we may of course
name a state as ``final'' and imagine that we read the result of the
computation (whatever this means) when the QTM enters such a state,
but we must cope with the fact that the computation will go on
afterwards. Moreover, since any computation of a QTM must be
reversible~\cite{Ben73}, in the sense that the operator describing the evolution of
the QTM must be invertible, we cannot neither force the machine to
loop on its final configuration. On the other hand, because of
reversibility, even the initial configuration must have a predecessor.
Summing up, an immediate consequence of all the above considerations
is that every state must have at least one incoming and one outgoing
transition and that such transitions must accord to several
constraints forced by quantum physics principles.  In particular, transitions must enter the initial state---since a priori
it might be reached as the evolution from a preceding configuration
--- and exit the final state---allowing the machine configuration to
evolve even after it has reached the final result.

The reversibility physical constraints are technically expressed by the
requirement that the \emph{evolution operator} of a QTM be
\emph{unitary}.  If we now want to use a QTM to compute some result,
we are still faced with the problem of when (and how, but for the
moment let's postpone this) reading such a final result, given that
the computation evolves also from the final state, and that, without
further constraints, it might of course evolve in many, different possible ways.
Bernstein and Vazirani in their seminal paper~\cite{BerVa97} (from now
on we shall refer to this paper as ``\BeV'') first define (non
unitary) QTMs; select then the ``well-formed'' QTMs as the
unitary-operator ones; and define finally ``well-behaved'' QTMs as
those which produce a superposition in which all classical
configurations are simultaneously (and for the first time) in the
final state. What happens after this moment, it is not the concern of
\BeV.

Our goal is to relax the requirement of simultaneous ``termination'',
allowing meaningful computations to reach superpositions in which some
classical configurations are final (and give a ``result''), and some
are not. Those which are not final, should be allowed to continue the
computation, possibly reaching a final state later.  The ``final
result'' will then be defined as a limit of this process.
In order to assure that at every step of the
computation the superposition of the final configurations is a valid
approximation of the limit result, we must guarantee that once a final
state is entered, the ``result'' that we read off is not altered by
the further (necessary, by unitarity) evolution. To obtain this, we
restrict the transition function out of a final state, without
violating unitarity. We obtain this goal by ``marking'' the symbols on
the tape---once a final state is entered, the machine remains in
that final state, replaces the symbol $a$ under the head with the same
marked symbol $\symextra{a}$, and moves to the right. Dually (to
preserve unitarity), when the machine is in an initial state, it
may rollback to another initial configuration in which the symbol $a$
to the left of the head is replaced by the marked symbol
$\symextra{a}$ and the head is on it; that is, looking at the
corresponding forward transition, when in the initial state, if the
machine head reads a marked $\symextra{a}$, then the machine remains
in that initial state, replaces the symbol with the unmarked $a$, and
moves to the right. The role of the extra marked symbols is restricted
to these ``final'' and  ``initial'' evolutions. In particular, there are no transitions out of a final
state when reading a marked symbol, or which enter an initial state
writing a marked symbol, and no transitions at all---involving marked symbols---entering or exiting
a state that is neither final nor initial. That these machines
correctly induce a unitary operator is the content of
Theorem~\ref{thm:time-evol-unitary}. In the following, marked symbols
will be called the \emph{extra symbols} and we will generalise this
discussion from a final (initial) state to a set of \emph{target}
(\emph{source}) states.

After the definition of QTMs and of the corresponding functions, we
will discuss their expressive power, comparing them to the QTMs
studied in the literature.  The QTMs of \BeV form a robust class, but
meaningful computations are defined only for classical inputs (a
single natural number with amplitude 1). Moreover, their QTMs
``terminate'' synchronously---either all paths in superpositions
enter a final state at the same time, or all of them diverge. As a
consequence, there is no chance to study---and give meaning---to
infinite computations.  More important, the class of ``sensible'' QTMs
(in \BeV's terminology: the well-formed, well-behaved, normal form
QTMs) is not recursively enumerable, since the constraint of
``simultaneous termination'' is undecidable. 

In Deutsch's original
proposal~\cite{Deu85}, any quantum TM has an explicit termination bit
which the machine sets when entering a final configuration. While it
is guaranteed that final probabilities are preserved, the observation
protocol requires that one checks termination at every step,
since the machine may well leave the final state (and change the tape).
Deutsch's machines could in principle be used to
define meaningful infinite computations, but we know of no such an
attempt.

In our analysis: (i) there is no termination bit: a quantum
configuration is a genuine superposition of classical configurations;
(ii) any computation path (that is, any element of the superposition)
evolves independently from the others: any path may terminate at its
own time, or may diverge; (iii) infinite computations are meaningful;
(iv) we may observe the result of the computation in a way akin to
Deutsch's one, but with the guarantee that once a final state is entered,
the machine will not change the corresponding ``result'' during a subsequent
computation;
(v) the class of QTMs is recursively enumerable, thus opening the
door to a quantum computability theory which may follow some of the
classical developments.

\section{Quantum Turing Machines}

In this section we define quantum Turing machines. We assume the
reader be familiar with classical Turing machines (in case,
see~\cite{Davis58}).

In all the paper we shall assume that the \emph{tape alphabet} $\Sigma$ is
finite and contains at least the symbols $1$ and $\Box$: $1$ will be used
to code natural numbers in unary notation; $\Box$ will be the blank symbol.
For any $n\in\NN$, $\nstring{n}$ will denote the string $1^{n+1}$. The
greek letters $\alpha, \beta$, eventually indexed, will denote strings
in $\Sigma^*$; $\lambda$ will denote the empty string; 
$\alpha\beta$ will denote the concatenation of $\alpha$ and $\beta$.

\subsection{QTM} 

Given a tape alphabet $\Sigma$, we associate an \emph{extra tape
  symbol} $\symextra{a}$ to any $a\in\Sigma$ (including the blank
$\Box$); $\abextra{\Sigma} =\{\symextra{a} \mid a \in \Sigma\}$ is the
\emph{extra tape alphabet}.  Finally,
$\aball{\Sigma} = \Sigma \cup \abextra{\Sigma}$.  As we discussed in
the introduction, extra tape symbols will only appear in computations
involving initial or final states.

As in \BeV, we assume that a machine move always implies a displacement
of the machine head, to the left ($L$) or the right ($R$) on the tape;
$\DD=\{L,R\}$ is the set of the displacements.

For $I$ denumerable, $\ell^2(I)$ is the Hilbert space of square summable,
$I$-indexed sequences of complex numbers
  $$\left\{\phi:I \rightarrow\CC \mid \sum_{C\in I }|\phi(C)|^2 <
    \infty\right\}, $$ 
equipped with an \emph{inner product}
  $\kinprod{.}{.}$ and the
  \emph{euclidean norm} $\norm{\phi}=\sqrt{\kinprod{\phi}{\phi}}$. See
  Appendix~\ref{sec:HS} for more details.

\begin{definition}[Quantum Turing machines]\label{def:QTM}
  Given a finite set of states $\Q$ and an alphabet $\Sigma$, a Quantum
  Turing Machine (QTM) is a tuple
  $$M=\langle \Sigma, \Q, \Q_s, \Q_t, \delta, q_i, q_f\rangle$$
  where
  \begin{itemize}
  \item $\Q_s\subseteq \Q$ is the set of \emph{source states} of $M$, and
    $q_i\in \Q_s$ is a distinguished source state named the \emph{initial
      state} of $M$;
  \item $\Q_t\subseteq \Q$ is the set of \emph{target states} of $M$, and
     $q_f \in \Q_t$ is a distinguished target state named the
    \emph{final state} of $M$;
  \item if we define $\Q_0 = \Q \setminus (\Q_s\cup \Q_t)$, then
    $\delta=\delta_0\cup\delta_s\cup\delta_t$ is the \emph{quantum
      transition function} of $M$, defined as the union of the
    following three functions with disjoint domains:
    \begin{align*}
      \delta_0 &: ((\Q_0\cup \Q_s)\times\Sigma) \to \ell^2((\Q_0\cup \Q_t) \times \Sigma \times \DD) \\
      \delta_s &: (\Q_s\times\abextra{\Sigma}) \to \ell^2(\Q_s \times \Sigma \times \DD) \\
      \delta_t &: (\Q_t\times\Sigma) \to \ell^2(\Q_t \times \abextra{\Sigma} \times \DD)
    \end{align*}
  \item the \emph{source transition function} $\delta_s$ is defined  by
      $$\delta_s(q_s,\symextra{a})(q_s',b,d) = 
      \begin{cases}
      1 & \mbox{if $(q_s',b,d) = (q_s,a,R)$} \\
      0 & \mbox{otherwise}
    \end{cases}$$ for every $(q_s,\symextra{a})\in \Q_s\times\abextra{\Sigma}$ and every
    $(q_s',b,d)\in \Q_s\times\Sigma\times\DD$
  \item the \emph{target transition function} $\delta_t$ is defined by
    $$\delta_t(q_t,a)(q_t',\symextra{b},d) =
    \begin{cases}
      1 & \mbox{if $(q_t',\symextra{b},d) = (q_t,\symextra{a},R)$} \\
      0 & \mbox{otherwise}
    \end{cases}$$
    for every $(q_t,a)\in \Q\times\Sigma$ and every
    $(q_t',\symextra{b},d)\in \Q_t\times\abextra{\Sigma}\times\DD$
  \item the \emph{main transition function} $\delta_0$ 
    satisfies the following \emph{local unitary conditions} 
    \begin{enumerate}
    \item for any $(q,a)\in (\Q_0\cup \Q_s) \times \Sigma$
      $$\sum_{(p,b,d)\in (\Q_0\cup \Q_t)\times\Sigma\times\DD}
      |\delta_0(q,a)(p,b,d)|^2=1$$ 
    \item for any $(q,a), (q',a')\in (\Q_0\cup \Q_s)\times\Sigma$ with $(q,a)\neq(q',a')$\\
      $${\sum_{(p,b,d)\in (\Q_0\cup \Q_t)\times\Sigma\times\DD} \delta_0(q',a')(p,b,d)^*\delta_0(q,a)(p,b,d)=0}$$
    \item for any  $(q,a,b), (q',a',b')\in (\Q_0\cup \Q_s)\times\Sigma^2$\\
      $${\sum_{p\in (\Q_0\cup \Q_t)} \delta_0(q',a')(p,b',L)^*\delta_0(q,a)(p,b,R)=0}$$
    \end{enumerate}
  \end{itemize}
\end{definition}

We remark that the domains and codomains of the three transition
functions $\delta_0$, $\delta_s$ and $\delta_t$ are  disjoint. Indeed,
if we define
\begin{align*}
  \varS_0 &=  (\Q_0 \cup \Q_s) \times \Sigma 
  & \varS_s &= \Q_s \times \abextra{\Sigma}
  & \varS_t &=  \Q_t \times \Sigma \\
  \T_0 &=  (\Q_0 \cup \Q_t) \times \Sigma 
  & \T_s &= \Q_s \times \Sigma
  & \T_t &=  \Q_t \times \abextra{\Sigma}
\end{align*}
we can then write, in a compact way
$$\delta_x =\varS_x \to  \ell^2(\T_x \times \DD)
\qquad\qquad 
\mbox{for $x\in\{0,s,t\}$}$$
and see that,  for $x,y\in\{0,s,t\}$,
$$ \varS_x \cap \varS_y = \emptyset
\qquad\qquad\qquad
 \T_x \cap \T_y = \emptyset$$
when $x\neq y$.

% As usual, we shall sometimes write $\delta(q,a,p,b,d)$ for
% $\delta(q,a)(p,b,d)$ (the amplitude of $(p,b,d)$ in $\delta(q,a)$).

\subsection{Configurations}
A configuration of a (classical) TM is a triple formed  by
the content of the tape, the state of the machine and the position of
the tape head. As usual, we assume that only a finite portion of the
tape contains non-blank symbols. We may therefore represent such a configuration
as a triple
$\langle\alpha,q,\beta\rangle \in \aball{\Sigma}^*\times \Q\times
\aball{\Sigma}^*$ where:
\begin{enumerate} 
\item $q$ is the current state; 
\item $\beta$ is the right content of the tape and its first symbol is
  the one under the head (we say also: in the current cell).  That is,
  $\beta=u\beta'$, where the \emph{current symbol} $u$ is the content
  of the \emph{current cell} (i.e.\ the one pointed by the tape head)
  and $\beta'$ is the longest string on the tape ending with a symbol
  different from $\Box$ and whose first symbol (if any) is written in
  the cell immediately to the right of the current cell; by
  convention, when the current cell and all the right content of the
  tape is empty, we shall also write
  $\langle \alpha, q, \lambda\rangle$ instead of
  $\langle \alpha, q, \Box \rangle$; 
\item $\alpha$ is the left content of the tape. That is, it is either
  the empty string $\lambda$, or it is the longest string %in $\Sigma^*$
  on the tape starting with a symbol different from $\Box$, and whose
  last symbol is written in the cell immediately to the left of the
  current cell.
\end{enumerate}

According to this definition, in a configuration
$\langle \alpha, q, \beta\rangle$ the string $\alpha$ does not start
with a blank $\Box$, and $\beta$ does not end with a blank.  In the
following it will be useful to manipulate configurations which are
extended with blanks to the right (of the right content) or to the
left (of the left content). For this, we equate configurations up to
the three equivalence relations induced by the following equations
\begin{gather*}
  \alpha \simeq_l \Box\alpha
  \qquad\qquad\beta \simeq_r \beta\Box\\
  \langle \alpha, q, \beta \rangle \simeq 
  \langle \alpha', q, \beta' \rangle \qquad\qquad 
  \mbox{when $\alpha\simeq_l\alpha'$ 
    and $\beta\simeq_r\beta'$}
\end{gather*}

We now turn to QTMs. Observe first that while cells containing the
blank symbol $\Box$ are considered empty, cells containing the extra
symbol $\symextra{\Box}$ are not empty and should not be ignored on
the left/right side of the tape.  Moreover, in view of the particular
evolution required for source/target states, and the special role of
extra symbols, some of the triples $\langle \alpha, q, \beta\rangle$
cannot occur as actual configurations in a computation of a QTM. We
limit our QTMs to configurations where extra symbols appear in
$\alpha\beta$ only when the current state is a source (target) state
and, moreover all the extra symbols are immediately to the right
(left) of the tape head.

\begin{definition}[configurations]~\label{Def:config}
  Let $M=(\Sigma, \Q, \Q_s, \Q_t,\delta, q_0, q_f)$ be a QTM. A \emph{configuration} of $M$ is a triple
  $\langle\alpha,q,\beta\rangle \in 
  \aball{\Sigma}^*\times \Q\times\aball{\Sigma}^*$ s.t.
  \begin{enumerate}
  \item if $q \not\in \Q_s \cup \Q_t$, then $\alpha\beta\in\Sigma^*$,
    that is, the tape does not contain extra symbols;
  \item if $q \in  \Q_s$, then $\alpha\in\Sigma^*$ and
    $\beta\in\abextra{\Sigma}^*\Sigma^*$;
  \item dually, if $q \in \Q_t$, then $\beta\in\Sigma^*$ and
    $\alpha\in\Sigma^*\abextra{\Sigma}^*$.
  \end{enumerate}
  The set of configurations of $M$ is denoted by $\GConf_M$. A
  configuration of $M$ is a \emph{source}/\emph{target configuration}
  when the corresponding state is a source/target state, moreover, it
  is a \emph{final}/\emph{initial configuration} when the current
  state is final/initial. By $\FConf_M$ we shall denote the set of the
  final configurations of $M$.
\end{definition}

In the following, the index $M$ in $\GConf_M$ and in the other names
indexed by the machine might drop when clear from the context.

%\begin{remark}
%Since quantum computations must be reversible, a QTM cannot have
%genuine finite computations. The final condition~\ref{delta-2} of
%Definition~\ref{def:preQTM} takes care of this, guaranteeing that once
%a classical final configuration is reached in a superposition, than on
%that configuration the QTM will simply move the control to the right,
%without exiting the final state or modifying the tape or the
%amplitude.
%\end{remark}

\subsection{Quantum configurations}
As already discussed in the introduction, the evolution of a QTM is
described by superpositions of configurations (as defined in
Definition~\ref{Def:config}).  If
$\B \subseteq \CConf_{M}$ is a set of configurations,
superpositions are elements of the Hilbert space $\ell^2(\B)$
(see, e.g.,~\cite{Con90,RomanBook}).  Quantum configurations of a QTM
$M$ are those elements of $\ell^2(\GConf_M)$ with unitary norm. We
remark that, since there is no bound on the size of the tape in a
configuration, the Hilbert space of the configurations must be
infinite dimensional.

\begin{definition}[quantum configurations]
  Let $M$ be a QTM.  The elements of the set
  $\qConf_M = \{\phi\in \ell^2(\GConf_M) \mid \sum_{C\in
    \GConf_M}|\phi(C)|^2 = 1\}$
  are the \emph{q-configurations} (quantum configurations) of $M$.
\end{definition} 

We shall use Dirac notation (see Appendix~\ref{sec:HS}) for the
elements $\phi,\psi$ of $\qConf_M$, writing them $\ket{\phi}, \ket{\psi}$.

\begin{definition}
  For any set of configurations
  $\B\subseteq \CConf_{M}$ and any $C\in\B$
  let $\ket{C}:\B \to \CC$ be the function
  \[
  \ket{C}(D)= \left\{
    \begin{array}{ll}
      1\;\;&\mbox{if}\;C=D\\
      0\;\;&\mbox{if}\;C\neq D.
    \end{array}
  \right. 
  \]
  The set $\cb{\B}$ of all such functions is a Hilbert basis
  for $\ell^2(\B)$ (see, e.g., \cite{NiOz02}). In particular,
  following the literature on quantum computing, $\cb{\GConf_M}$ is
  called the \emph{computational basis} of $\ell^2(\GConf)$.  Each
  element of the computational basis is called \emph{base
    q-configuration}.
\end{definition}

With a little abuse of language we shall write $\ket{C}\in\ket{\phi}$ when
$\phi(C)\neq 0$.
 The \emph{span of} $\cb{\B}$, denoted by
$\SPAN{\cb{\B}}$, is the set of the finite linear
combinations with complex coefficients of elements of
$\cb{\B}$; $\SPAN{\B}$ is not a Hilbert space,
although $\ell^2(\B)$ is the (unique, up to isomorphism)
\emph{completion} of $\SPAN{\B}$.  Moreover, each unitary
operator $U$ on $\SPAN{\B}$ has a unique 
unitary extension on $\ell^2(\B)$~\cite{BerVa97}.

For a list of the main definitions, properties and notations on
Hilbert spaces with denumerable basis, see Appendix~\ref{sec:HS}. In
particular, subsection~\ref{ssec:dirac-notation} presents a synoptic
table of the so-called Dirac notation that we shall use in the paper.

\subsection{Time evolution operator}

For any QTM $M$ with alphabet $\Sigma$, and  space of states $\Q$, the step
function
$\gamma_{M}: \CConf_{M} \times \Q \times \aball{\Sigma}
\times \DD \to \CConf_{M}$
is the map that, given a configuration of the tape and a triple
$(p, b, d)$ describing a ``classical'' step of a Turing
machine, replaces the symbol in the current cell with the symbol $p$,
moves the head on the $d$ direction, and sets the machine into the new
state $p$. Formally:
$$\gamma_{M}(\langle \alpha w, q, u \beta \rangle, p,v,d) \simeq
\begin{cases}
  \langle \alpha wv, p, \beta \rangle & \mbox{when $d = R$} \\
  \langle \alpha, p, wv\beta \rangle & 
  \mbox{when $d = L$.}
\end{cases}
$$
The evolution of a QTM $M=\langle \Q, \Sigma, \delta, q_0,
q_f\rangle$ can then be defined as a map on % ground
q-configurations. Following the three-parts definition of the transition function,
let
$$\SConf{x}_M = \{\langle\alpha, q, u \beta\rangle \in \GConf_M 
\mid (q,u) \in \varS_x \}
\qquad\qquad
\mbox{with $x\in\{0,s,t\}.$}
$$
It is easily seen that $\SConf{0}_M$, $\SConf{s}_M$, and $\SConf{t}_M$ are a
partition of $\GConf_M$ (they are pairwise disjoint, because $\varS_0$, $\varS_s$, and $\varS_t$ are pairwise disjoint). % and $\GConf_M = \SConf{0}_M \cup \SConf{s}_M \cup \SConf{t}_M$. 
Therefore,
given
$$C = \langle \alpha, q, u\beta\rangle \in \GConf_M \qquad\qquad 
C_{p,v,d}=\gamma_{M}(C,p,v,d)$$
we can define
$$W_M(\ket{C}) = \sum_{(p,v,d)\in \T_x\times\DD}
\delta_x(q,u)(p,v,d)\,\ket{C_{p,v,d}}
\qquad \mbox{when $C \in \SConf{x}_M$.}
$$

\begin{proposition}\label{prop:WM-welldef}
  $W_M(\ket{C}) \in \SPAN{\cb{\GConf}}$, for any
  $C\in\GConf$. Then, $W_M$ naturally extends to an automorphism on
  the linear space of % ground 
  q-configurations
  $$W_M:\SPAN{\cb{\GConf}}\to\SPAN{\cb{\GConf}}.$$
%  Moreover, $W_M({\SPAN{\cb{\SConf{x}_M}}}) \subseteq \SPAN{\cb{\TConf{x}_M}}$.
\end{proposition}
\begin{proof}
  Let $C$ and $C_{p,v,d}$ be as in the definition of $W_M$.
  \begin{enumerate}
  \item Let $C\in\SConf{0}_M$. If $(p,v)\in \T_0$, then
    $C_{p,v,d}\in\SConf{t}_M$ when $p\in \Q_t$, and
    $C_{p,v,d}\in\SConf{0}_M$ otherwise. Thus,
    $C_{p,v,d}\in\GConf_M$, for every $(p,v,d)\in
    \T_0\times\DD$.

  \item Let $C\in\SConf{s}_M$, that is, $u=\symextra{a}$ for some
    $a\in\Sigma$ and $C
    =\langle\alpha,q,\symextra{a}\symextra{\gamma}\beta\rangle$, where
    $\alpha,\beta,\gamma\in\Sigma^*$. If $(p,v)\in \T_s$, then
    $\delta(q,\symextra{a})(p,v,d)\neq 0$ iff $v=a$ and $d=R$, namely,
    $W_M(\ket{C})=\ket{C_{p,a,R}}$ with $C_{p,a,R}=\langle\alpha
    a,p,\symextra{\gamma}\beta\rangle$. Then, if the symbol
    $\symextra{a}$ replaced by $a$ was the last extra symbol on the
    tape, that is $\gamma=\lambda$, then $C_{p,a,R}\in\SConf{0}_M$,
    otherwise $C_{p,a,R}\in\SConf{s}_M$. In any case, $W_M(\ket{C})\in
    \SPAN{\cb{\GConf}}$.

  \item Let $C\in\SConf{t}_M$, that is, $u=a$ for some $a\in\Sigma$ and
    $C =\langle\alpha\symextra{\gamma},q,\beta\rangle$, where
    $\alpha,\beta,\gamma\in\Sigma^*$. If $(p,v)\in \T_t$, then
    $\delta(q,a)(p,v,d)\neq 0$ iff $v=\symextra{a}$ and $d=R$; namely,
    $W_M(\ket{C})=\ket{C_{p,\symextra{a},R}}$ with
    $C_{p,\symextra{a},R}=\langle\alpha \symextra{\gamma}\symextra{a},
    p, \beta\rangle$. Then, $C_{p,a,R}\in\SConf{t}_M$ and
    $W_M(\ket{C})\in \SPAN{\cb{\GConf}}$.
  \end{enumerate}
  
  Summing up, $W_M(\ket{C})\in \SPAN{\cb{\GConf}}$ in any case. Thus,
  $W_M$ uniquely extends to an automorphism on $\SPAN{\cb{\GConf}}$ by
  linearity: that is 
  $$W_M(\sum_{C\in\GConf_M} k_C \ket{C}) = \sum_{C\in\GConf_M} k_C
  W_M(\ket{C}).$$
\end{proof}

By completion, $W_M$  extends in a unique way to an operator
on the Hilbert space of q-configurations.

\begin{definition}[time evolution operator]\label{Def:TimeEvolutionOp}
  The \emph{time evolution operator} of $M$ is the unique extension
  $$U_M:\ell^2(\GConf_M)\to \ell^2(\GConf_M)$$
  of the linear operator
  $W_M:\SPAN{\cb{\GConf}}\to\SPAN{\cb{\GConf}}$.
\end{definition}

\begin{theorem}\label{thm:time-evol-unitary}
  The time evolution operator of a QTM is unitary.
\end{theorem}
\begin{proof}
  The proof is a variant of the one given by~\BeV and by Nishimura and
  Ozawa~\cite{NishOza10}.  In particular, we prove first that $U_M$ is
  an isometry of $\ell^2(\GConf)$, and then that, in this particular
  case, this implies that $U_M$ is unitary (which, in general, holds
  for finite dimensional Hilbert spaces only---in the infinite
  dimensional case an isometry might not be surjective).  The full
  details of the proof are given in
  Appendix~\ref{sec:proof-local-cond}.
\end{proof}

In Appendix~\ref{sec:proof-local-cond} we shall not only show
that the unitary local conditions imply the unitarity of the time
evolution operator (i.e., Theorem~\ref{thm:time-evol-unitary}), but
that they are also necessary.  We remark that, this is not just a
simple adaptation to our case of the already known proofs for \BeV
QTM; in fact, we also simplify the argument that allows to show that
the isometry of $U_M$ implies its unitarity.

Since the time evolution operator of a QTM is unitary, it preserves
the norm of its argument, hence it maps q-configurations into
q-configurations.

\begin{proposition}
  Let $M$ be a QTM. If $\ket{\phi}\in\qConf$, then
  $U_M\ket{\phi}\in\qConf$.
\end{proposition}

\begin{definition}[initial and final configurations]
  A q-configuration $\ket{\phi}=\sum_i e_i\ket{C_i}$ is \emph{initial}
  if all the $C_i$ are initial, and it is \emph{final} if all the
  $C_i$ are final.  Moreover, $\qFConf$ is the set of final
  q-configurations, and we shall denote by $\ket{\nstring{n}}$ the
  initial configuration $\langle\lambda,q_0,\nstring{n}\rangle$.
\end{definition}

\begin{definition}[computations]
  Let $M$ be a QTM and let $U_M$ be its time evolution
  operator.  For an initial q-configuration $\ket{\phi}$, the
  \emph{computation} of $M$ on $\ket{\phi}$ is the denumerable sequence
  $\{\ket{\phi_i}\}_{i\in \NN}$ s.t. 
  \begin{enumerate}
  \item $\ket{\phi_0}=\ket{\phi}$;
  \item $\ket{\phi_{i}}=U_M^i\ket{\phi}$.
  \end{enumerate}
\end{definition}

Clearly, any computation of a QTM $M$ is univocally determined by its
initial q-configuration. The computation of $M$ on initial
q-configuration $\ket{\phi}$ will be denoted by $K_{\ket{\phi}}^M$.

\begin{lemma}\label{lem:pers-final-conf}
  Let $C\in\FConf$ be a final configuration without extra symbols; in
  particular, let $C=\langle \beta, q_f,\alpha\rangle$ with
  $\alpha\beta\in\Sigma^*$ and $\alpha=a_1\ldots a_l$. For every
  $k\geq 0$, let
  $$C[k]= 
  \begin{cases}
    \langle \beta, q_f,\symextra{a}_1\ldots\symextra{a}_{k}a_{k+1}\ldots a_l\rangle
    & \quad \mbox{for } 0 \leq k < l\\
    \langle \beta, q_f,\symextra{a}_1\ldots\symextra{a}_l\symextra{\Box}^{k-l} \rangle
    & \quad \mbox{otherwise}
  \end{cases}$$
  that is, for $k>0$, $C[k]$ is obtained from $C$ by replacing the
  current symbols and the first $k-1$ symbols to its right by the
  corresponding extra symbols.
  \begin{enumerate}
  \item $\ket{C[j]} = U_M^{j-i}\ket{C[i]}$, for every $i,j\geq
    0$. More generally, for every $i,j\geq 0$ and every $\phi\in\qConf$,
    $$\inprod{U^{j-i}_M \ket{\phi}}{\ket{C[j]}} = \kinprod{\phi}{C[i]}$$
  \item $U_M^{-j-i}\ket{C[i]}\in\ell^2(\GConf\setminus\FConf)$, for
    every $i\geq 0$ and $j>0$. 
  \item Let $\ket{\phi}\in\ell^2(\GConf\setminus\FConf)$. For every
    $i,j\geq 0$, we have that:
    \begin{enumerate}
    \item $\ket{C[i]}\in U^j\ket{\phi}$ only if $i < j$.
    \item $\ket{C[i]}\not\in U^{-j}\ket{\phi}$, that is,
      $U^{-j}\ket{\phi}\in\ell^2(\GConf\setminus\FConf)$
    \end{enumerate}
  \end{enumerate}
\end{lemma}
\begin{proof}
  \mbox{}
  \begin{enumerate}
  \item By the definition of the transition function, we can easily
    see that $\ket{C[i+1]}=U_M\ket{C[i]}$, for every $i\geq
    0$. Therefore,
    $\ket{C[j]}=U_M^j\ket{C[0]}=U_M^{j-i}U_M^i\ket{C[0]}=U_M^{j-i}\ket{C[i]}$,
    for every $i,j\geq 0$. Then, since $U_M$ is unitary,
    $\inprod{U^{j-i}_M \ket{\phi}}{\ket{C[j]}} = \tek{\phi} U^{i-j}_M
    \ket{C[j]} = \kinprod{\phi}{C[i]}$.
  \item Let $D$ be any final configuration without extra symbols. By
    the previous item,
    $\inprod{U^{-j-i}_M \ket{C[i]}}{\ket{D[k]}} =
    \kinprod{C[i]}{D[k+j+i]}=0$,
    for every $k\geq 0$, since $j>0$ and then $k+j+i>i$.  Therefore,
    since $D[k]\not\in U^{-j-i}_M \ket{C[i]}$, for any $D$ and any
    $k$, we conclude that
    $U_M^{-j-i}\ket{C[i]}\in\ell^2(\GConf\setminus\FConf)$.
  \item They are immediate consequences of the previous two
    items. Indeed,
    $\inprod{U^j_M \ket{\phi}}{\ket{C[i]}} =
    \kinprod{\phi}{C[i-j]]}=0$
    for $i \geq j$, and
    $\inprod{U^{-j}_M \ket{\phi}}{\ket{C[i]}} =
    \kinprod{\phi}{C[i+j]]}=0$, for $j\geq 0$,
since
    $\ket{\phi}\in\ell^2(\GConf\setminus\FConf)$.
  \end{enumerate}
\end{proof}

\begin{remark}\label{rem:final-stable}
  The previous lemma shows that the final configurations reached along
  a computation are stable and do not interfere with other branches
  of the computation in superposition, which may enter into a final
  configuration later. Indeed, given a configuration
  $\ket{\phi} = \ket{\phi_f} +\ket{\phi_{nf}}$, where
  $\ket{\phi_f}\in\qFConf$ and $\ket{\phi_{nf}}$ does not contain any
  final configuration, let
  $\psi=U^i\ket{\phi}=U^i \ket{\phi_f} + U^i\ket{\phi_{nf}}$.  Any
  final configuration in $U^i\ket{\phi_{nf}}$ contains less than $i$
  extra symbols, while any final configuration $C[k]$ with $k$ extra
  symbols in $\ket{\phi}$ is mapped into a configuration $C[i+k]$
  with $i+k$ extra symbols, without changing the value associated to
  the configuration, since $\val(C[k])=\val(C[k+i])$. Moreover, $C[k]$ and
  $C[i+k]$ have the same coefficient in $\ket{\phi}$ and $\ket{\psi}$,
  respectively, since
  $\kinprod{\psi}{C[i+k]}=\inprod{U^i\ket{\phi_f}}{\ket{C[i+k]}}=
  \kinprod{\phi}{C[k]}$.
\end{remark}

\subsection{A comparison with Bernstein and Vazirani's QTMs: part 1}

We refer to \BeV for the precise definitions of the QTMs used in that
paper. For the sake of readability, we informally recall the notion of
what they call \textit{well formed, stationary, normal form QTMs}
(B\&V-QTMs in the following).

%memo: stationary = well-behaved + each configuration has the tape head in the start cell

A B\&V-QTM $M=\langle \Sigma, \Q, \delta, q_0, q_f\rangle$ is defined
as our QTM (with one source state and one target state only) with the
following differences:
\begin{enumerate}
\item the set of configurations coincides with all possible classical
  configurations, namely all the set $\Sigma^*\times \Q\times\Sigma^*$.
\item no superposition is allowed in the initial q-configuration (it
  must be a classical configuration $\langle\alpha,q,\beta\rangle$
  with amplitude 1);
\item let $\ket{C}$ be such an initial configuration and let\\
  \centerline{%
    $k=\min\{j \mid \mbox{$U_{M}^{j}\ket{C}$ contains a final
      configuration}\}$%
  } %
  If such a $k$ exists, then (i) all the configurations in
  ${U_M}^k\ket{C}$ are final; (ii) for all $i < k, {U_M}^i\ket{C}$
  does not contain any final configuration. We say in this case that
  the QTM halts in $k$ steps in $U_M^k\ket{C}$;
\item if a QTM halts, then the tape head is on the start cell of the
  initial configuration;
\item there are no extra symbols and the transitions out of the final
  state or into the initial state are replaced by loops from the final
  state into the initial state, that is, $\delta(q_f,a)(q_0,a,R) = 1$
  for every $a\in\Sigma$. Therefore, because of the local unitary
  conditions, that must hold in the final state also, these are the
  only outgoing transitions from the final state and the only incoming
  state into the initial state, that is, $\delta(q_f,a)(q',a',d) = 0$ if
  $(q',a',d)\neq(q_0,a,R)$ and $\delta(q',a')(q_0,a,d)=0$ if
  $(q',a',d) \neq (q_f,a,R)$.
\end{enumerate}

\begin{theorem}\label{theor:fromBeVtoOur}
  For any B\&V-QTM $M$ there is a  QTM $M'$ s.t.\ for each initial configuration
  $\ket{C}$, if $M$ with input $\ket{C}$ halts in $k$ steps in a final configuration
  $\ket{\phi}=U_M^k\ket{C}$, then $U_{M'}^k{\ket{C}}=\ket{\phi}$.
\end{theorem}
\begin{proof}
  The QTM $M'$ has the same states of $M$, only one source state, the
  initial state $q_0$, and only one target state, the final state
  $q_f$. Therefore, if $M=\langle \Sigma, \Q, \delta, q_0, q_f\rangle$,
  we take
  $M'=\langle \Sigma, \Q, \{q_0\}, \{q_f\}, \delta', q_0, q_f\rangle$.

  The source part $\delta_s'$ and the target part $\delta_t'$ of the
  transition function $\delta'$ of $M'$ are uniquely determined by the
  definition of QTM. The function $\delta_0'$ is instead the
  restriction of $\delta$ to the domain
  $\varS_0=\Q\setminus\{q_f\}\times\Sigma$, that is, for each $q\neq q_f$
  and $a\in\Sigma$, we have
  $\delta_0'(q,a)(p,b,d)=\delta(q,a)(p,b,d)$, for every
  $(p,b,d)\in \T_0\times \DD$. The unitary local conditions hold for
  $\delta_0'$, since they hold for $\delta$ and because, as already
  remarked, in a \BeV-QTM, if $q\neq q_f$, then $\delta(q,a)(q_0,b,d)=0$,
  for every $a,b\in\Sigma$ and $d\in\DD$.
  
  By construction, it is clear that for each $i\leq k$, s.t.\
  $\ket{\phi_i}=U_{M'}^i{\ket{C}}$ is not final, then $\ket{\phi_i}=U_{M}^i{\ket{C}}$.
\end{proof}

\section{Quantum Computable Functions}

In this section we address the problem of defining  the concept of quantum computable function in an ``ideal'' way, without taking into account any measurement protocol. The problem of the observation protocol will be addressed in Section~\ref{Sect:observables}.
Here we show how each QTM  naturally defines a computable function from the sphere of radius $1$ in $\ell^2$ to the set of (partial) probability distributions on the set of natural numbers.

\begin{definition}[Probability distributions]
  \mbox{}
  \begin{enumerate}
  \item A \emph{partial probability distribution} (PPD) of natural
    numbers is a function $\pd{}: \NN \to \mathbb{R}_{[0,1]}$ such
    that $\sum_{n\in\NN} \pd{}(n)\leq 1$.
  \item  If $\sum_{n\in\NN} \pd{}(n)= 1$, $\pd{}$ is a \emph{probability
      distribution} (PD).
  \item $\PP$ and $\PP_1$ denotes the sets of all
    the PPDs and PDs, respectively.
  \item If the set $\{n : \pd{}(n)\neq 0\}$ is finite, $\pd{}$ is
    \emph{finite}.
  \item Let $\pd{}', \pd{}''$ be two PPDs, we say that
    $\pd{}' \leq \pd{}''$ ($\pd{}' < \pd{}''$) iff for each
    $n\in \NN$, $\pd{}'(n)\leq \pd{}''(n)$ ($\pd{}'(n) < \pd{}''(n)$).
  \item Let $\varP=\{\pd{i}\}_{i\in\NN}$ be a denumerable
    sequence of PPDs; $\varP$ is \emph{monotone} iff
    $\pd{i}\leq \pd{j}$, for each $i<j$.
  \end{enumerate}
\end{definition}

\begin{remark}
  In the following, we shall also use the notation
  $\pd{}(\bot)=1-\sum_{n\in\NN}\pd{}(n)$. By definition,
  $0\leq \pd{}(\bot) \leq 1$, and a PPD is a PD iff $\pd{}(\bot)=0$.
\end{remark}
% \fbox{\blue{Dove usata la seguente notazione?}}

% Sometimes it will be convenient to denote a partial probability distribution $\pd{}$ as $\{p_0:x_0,p_1:x_1,\ldots\}$, where $p_i$ is $\pd{}(x_i)$ and only non zero probabilities are listed.

Since real numbers are a complete space, we have the following result.
\begin{proposition}
  Each $\varP\subseteq \PP$ has a supremum, denoted by $\lub \varP$.
\end{proposition}
\begin{proof}
  For each $n\in\NN$, the set
  $\varP_n=\{\pd{}(n) : \pd{}\in\varP\} $ has a supremum
  $\lub \varP_n$. It is a trivial exercise to verify that
  $(\lub \varP)(n)=\lub \varP_n$
  is indeed the supremum of $\varP$.
\end{proof}

We can now introduce the notion of limit of a sequence
$\varP=\{\pd{i}\}_{i\in\NN}$.

\begin{definition}
  Let $\varP=\{\pd{i}\}_{i<\NN}$ be a sequence of PPDs.  If
  for each $n\in\NN$ there exists $l_n=\lim_{i\to \infty} \pd{i}(n)$,
  we say that $\lim_{i\to \infty} \pd{i}=\pd{}$, with $\pd{}(n)=l_n$.
\end{definition}

The computed outputs of a QTM will  be defined as the limit of the
sequence of partial probability distributions obtained along its computations.

\begin{definition}[probability and q-configurations]\label{Def:compandppds}
  Given a configuration $C=\langle\alpha,q,\beta\rangle$, let
  $\val[C]$ be the number of $1$ and $\bar{1}$ symbols in
  $\alpha\beta$.
 \begin{enumerate}
 \item To any q-configuration $\ket{\phi}=\sum_C e_C \ket{C}$, we
   associate the partial probability distribution $\pd{\ket{\phi}}$
   s.t.
   $\prob{\ket{\phi}}{n}=\sum_{C\in\FConf, \val[C] =n} |e_C|^2$.
 \item For any computation
   $K_{\ket{\phi}}^M=\{\ket{\phi_i}\}_{i\in \NN}$, let
   $\pd{K_{\ket{\phi}}^M}$ be the sequence of PPDs
   $\{\pd{\ket{\phi_i}}\}_{i\in \NN}$.
  \end{enumerate}
\end{definition}

We  now show the key property that a QTM computation yields monotone
sequences of PPDs.  In its simple proof we see at work all the
constraints on the transition function of a QTM.  First, once in a
target state the machine can only change a (normal) symbol $a$ into
the corresponding extra symbol $\bar{a}$; as a consequence, the $\val$
of these configurations does not change.  Second, when entering for
the first time into a target state there are no extra symbols on the
tape, for the form of the transition function $\delta_0$, which is
reflected in the constraints on the configurations
(Definition~\ref{Def:config}). Finally, in a final configuration the
number of extra symbols counts the steps the machine performed since
entering for the first time into a target state.  These last two
properties defuse quantum interference between final configurations
reached in a different number of steps.

\begin{theorem}[monotonicity of computations]\label{theor:monot}
  For any computation $K_{\ket{\phi}}^M$ of a QTM $M$, the sequence of
  PPDs $\pd{K_{\ket{\phi}}^M}$ is monotone.
\end{theorem}

\begin{proof}
  It is a direct consequence of the properties already remarked in
  Lemma~\ref{lem:pers-final-conf} and
  Remark~\ref{rem:final-stable}. Anyhow, let us see a direct proof in
  details.

  \newcommand{\niu}[1]{{\small \maltese}#1} %
  Let $U$ be the time evolution operator of $M$.  We prove that
  $\forall i$ $\forall n$:
  $\prob{\ket{\phi_i}}{n}\leq \prob{U\ket{\phi_{i}}}{n}$.  Let us
  split any q-configuration in two parts, the sums of the final and of
  the non-final configurations:
  \[
  \ket{\phi_i}=\sum_{C\in \F_i} d_C\ket{C}
  +
  \sum_{D\not\in\FConf} d_D\ket{D}
  \]
  where $\F_i\subset \FConf$ are the final configurations in
  $\ket{\phi_i}$ with non-null amplitude. By applying $U$ we get
  \[
  U\ket{\phi_i}=\sum_{C\in\F_i} d_C U\ket{C}
  +
  \sum_{B\in \F_i'} d_B'\ket{B}
  +
  \sum_{A\not\in\FConf} d_A'\ket{A}
  \]
  where $\F_i'\subseteq\FConf$, and $d'_B\neq 0$ for $B\in\F_i'$, and
  $$
  U(\sum_{D\not\in\FConf} d_D\ket{D})=
  \sum_{B\in \F_i'} d_B'\ket{B}
  +
  \sum_{A\not\in\FConf} d_A'\ket{A}.
  $$
  The sum on non-final configurations does not contribute to
  $\pd{U\ket{\phi_{i}}}$. On the other hand, let
  $\niu{\ket{C}}=\#\{ \bar{a} : \bar{a}\in C\}$. 
  For each $C\in \F_i$:
  \begin{enumerate}
  \item  $U\ket{C} = \ket{C'} \in \FConf$; 
  \item  $\val[C]=\val[C']$;
  \item  $\niu{\ket{C'}}=\niu{\ket{C}}+1>0$;
  \item for any other $E\in \F_i$ (i.e., $E\neq C$), we have
    $U\ket{E}=E'\neq C'$.
  \end{enumerate} 
  As for the newly final configurations $B\in \F_i'$, we have
  $\niu{\ket{B}}=0$. Therefore, none of the $B\in \F_i$ is equal
  to any of the $C'$ s.t.\ $\ket{C'}= U\ket{C}$, and hence 
  $$\pd{U\ket{\phi_{i}}} 
       =
      \sum_{\dind{C\in \F_i}{\val[C]=n}} |d_C|^2
      +
      \sum_{\dind{B\in \F_i'}{\val[B]=n}} |d_B'|^2
      \  = \
      \pd{\ket{\phi_i}}
      +
      \sum_{\dind{B\in \F_i'}{\val[B]=n}} |d_B'|^2$$
% QUESTA E' LE VERSIONE DI ANDREA
%\[
%U(\ket{\phi_i})=\sum_{\ket{C}\in F_i,d_C\neq 0 } d_C U(\ket{C})
%+
%\sum_{\ket{B}\in\FConf-F_i, d_B\neq 0} d_B\ket{B}
%+
%\sum_{\ket{A}\not\in\FConf} d_A\ket{A}
%\]
%where:
%\begin{enumerate}
%\item $U(\sum_{\ket{D}\in\C-F_i, d_D\neq 0} d_D\ket{D})= \sum_{\ket{B}\in\FConf-F_i, d_B\neq 0} d_B\ket{B}
%+
%\sum_{\ket{A}\not\in\FConf} d_A\ket{A}$;
%\item for each $\ket{C}\in F_i$, $U(\ket{C})=\ket{C'}$   for some $C'\in F$;
%\end{enumerate}
%Observe that for each $\ket{C}\in F_i$:
%\begin{enumerate}
%\item $\niu{\ket{C}}\leq \niu{U(\ket{C})}=\niu{\ket{C}}+1\geq 1$;
%\item $\prob{\ket{C}}{n}=\prob{U\ket{C}}{n}$
%\end{enumerate}  
%and that 
%for each $\ket{C}\in F_i$ and for each $\ket{B}\in\FConf-F_i$, $\niu{U\ket{C}} > 
%\niu{\ket{B}}$
% and conclude.
\end{proof}

We now turn to the definition of the computed output of a QTM
computation.  The easy case is when a computation reaches a final
q-configuration $\ket{\psi}\in \FConf$ (meaning that all the
classical computations in superposition are ``terminated'')---in this
case the computed output is the PD $\pd{\ket{\psi}}$. Of course the
QTM will keep computing and transforming $\ket{\psi}$ into other
configurations, but these configurations all have the same PD.
However, we want to give meaning also to ``infinite'' computations,
which never reach a final q-configuration, yet producing some final
configurations in the superpositions.  In this case we define the
computed output as the limit of the PPDs yielded by the computation.

We need first the following lemma, whose proof is an easy consequence of the definition of limit for PPDs and of  well known properties of limits of real-valued, non-decreasing sequences on natural numbers.

\begin{lemma}\label{lemma:limMonot}
Let $K=\{\ket{\phi_i}\}_{i\in\NN}$ be a  monotone sequence of q-configurations, then the sequence of PPDs
$\varP=\{\pd{\ket{\phi_i}}\}_{i\in\NN}$ enjoys
$\lim_{i\to \infty}\pd{\ket{\phi_i}}= \lub\varP$.
\end{lemma}
From the lemma and Theorem~\ref{theor:monot} we have: 
\begin{corollary}\label{lemma:limComp}
Let $K=\{\ket{\phi_i}\}_{i\in\NN}$ be a computation, then the sequence 
$\{\pd{\ket{\phi_i}}\}_{i\in\NN}$ has limit  $\lim_{i\to \infty}\pd{\ket{\phi_i}}$ 
(denoted by $\lim K$).
\end{corollary}

The existence of the limit allows the following definition.
\begin{definition}[computed output of a QTM]
  The computed output of a QTM $M$ on the initial
  q-configuration $\ket{\phi}$ is the PPD
  $\pd{}=\lim K^M_{\ket{\phi}}$ (notation:
  $M_{\ket{\phi}}\to \pd{}$).     
% \fbox{\blue{I due item seguenti sono da cancellare, OK?}}
%  \item 
% If for each $\pd{}\in \PP$ it is not the case that $M_{\ket{\phi}}\to \pd{}$, we say that $M$
% on the initial q-configuration $\ket{\phi}$ \emph{diverges} (notation:
% $M_{\ket{\phi}}\uparrow$).
% \item If there exists a  $\pd{}\in\PP$ such that $M_{\ket{\phi}}\to \pd{}$, we say  that $M$
% on the initial q-configuration $\ket{\phi}$ \emph{converges} (notation:
% $M_{\ket{\phi}}\downarrow$).
% \end{enumerate}
\end{definition}

Let us note that a QTM always has a computed output.

%
%\fbox{\blue{mettiamo il seguente teorema??????????? direi di no... mi sembra inutile}}
%\blue{\begin{theorem}
%Let $M$ be a QTM, $\ket{\phi}$ be an initial q-configuration, and $K_{\ket{\phi}}^M=\{\ket{\phi_i}\}_{i\in \NN}$, the corresponding computation, then 
%there exists a PPD $\pd{}$ and  a denumerable sequence $\{\ket{\psi_i}\}_{i\in \NN}$ of final q-configurations\footnote{please note that $\{\ket{\psi_i}\}_{i\in \NN}$ is not a computation!}, with $\eval\;\ket{\psi_i}=\pd{}$ for any $i$, and
%$$
%\forall \epsilon\in \R^+\ \exists i \ s.t.\ \forall j > i
%\norma{\ket{\phi_j}- \ket{\psi_j}}\leq \epsilon.
%$$
%or in other words
%\[
%\lim_{j\to \infty} \norma{\ket{\phi_j} - \ket{\psi_j}} = 0
%\]
%\end{theorem}
%}

\begin{definition}
  Given a QTM $M$, a q-configuration $\ket{\phi}$ is \emph{finite} if
  it is an element of $\mathsf{span(CB(\GConf_M))}$.
  A computation $K_{\ket{\phi}}^M=\{\ket{\phi_i}\}_{i\in \NN}$ is
  \emph{finitary} with computed output $\pd{}$ if there exists a $k$
  s.t. $\ket{\phi_k}$ is final and $\pd{\ket{\phi_k}}=\pd{}$.
\end{definition}

\begin{proposition}\label{prop:partiality}
  Let $K_{\ket{\phi}}^M=\{\ket{\phi_i}\}_{i\in \NN}$ be a finitary
  computation with  computed output $\pd{}$. Then:
  \begin{enumerate}
  \item There exists a $k$, such that for each $j\geq k$, $\ket{\phi_j}$ is final 
  and $\pd{\ket{\phi_j}}=\pd{}$;
  \item $M_{\ket{\phi}}\to \pd{}$;
  \item $\pd{}$ is a PD. 
  \end{enumerate}
\end{proposition}
\begin{proof}
  Apply the definitions; $\pd{}\in\PP_1$ since $\ket{\phi_k}$ is
  final and has norm 1.
\end{proof}

With a computation $K_{\ket{\phi}}^M$, several cases may thus happen:
\begin{enumerate}
\item $K_{\ket{\phi}}^M$ is finitary. The output of the computation is a
  PD, which is determined after a finite number of steps;
\item $K_{\ket{\phi}}^M$ is not finitary, and
  $M_{\ket{\phi}}\to \pd{}\in\PP_1$. The output is determined
  as a limit;
\item $K_{\ket{\phi}}^M$ is not finitary, and
  $M_{\ket{\phi}}\to \pd{}\in\PP-\PP_1$ (the sum of the probabilities
  of observing natural numbers is $p<1$). Not only the result is
  determined as a limit, but we cannot extract a  PD from the
  output.
\end{enumerate} 

The first two cases above give rise to what
Definition~\ref{definition:qcf} calls a \emph{q-total} function.
Observe, however, that for an external observer, cases (2) and (3) are
in general indistinguishable, since at any finite stage of the
computation we may observe only a finite part of the computed output.

For some examples of QTMs and their computed output, see
Section~\ref{Sec:ExpressivePower}.

\subsection{Quantum partial computable functions}

We want our quantum computable functions to be defined over a natural
extension of the natural numbers.  Recall that, for any $n\in\NN$,
$\nstring{n}$ denotes the string $1^{n+1}$ and that
$\ket{\nstring{n}}=\ket{\langle\lambda,q_0,\nstring{n}\rangle}$. When
using a QTM for computing a function, we stipulate that initial
q-configurations are superpositions of initial classical
configurations of the shape $\ket{\nstring{n}}$. Such q-configurations
are naturally isomorphic to the space
$\ell^2_1= \left\{ \phi:\NN \rightarrow\CC \mid
  \sum_{n\in\NN}|\phi(n)|^2=1\right\}$
of square summable, denumerable sequences with unitary norm, under the
bijective mapping
$\nu(\sum d_k n_k) = \sum d_k
\ket{\nstring{n_k}}$.

% \fbox{\blue{Questa è la ``vecchia'' def:}}
% \begin{definition}[partial quantum computable functions]\label{definition:qcf}
%   \mbox{}
%   \begin{enumerate}
%   \item A partial function $f:\ell^2_1 \partialto \PP$ is
%     \emph{quantum computable} (q-computable) if there exists a QTM
%     $M$ s.t.
%     \begin{enumerate}
%     \item $f(\mathbf{x})=\pd{}$ iff $M_{\nu^{-1}(\mathbf{x})}\to \pd{}$;
%     \item $f(\mathbf{x})$ is defined (notation:
%       $f(\mathbf{x})\downarrow$) iff
%       $M_{\nu^{-1}(\mathbf{x})}\downarrow$;
%     \item $f(\mathbf{x})$ is undefined (notation:
%       $f(\mathbf{x})\uparrow$) iff $M_{\nu^{-1}(\mathbf{x})}\uparrow$;
%     \end{enumerate}
%   \item A q-partial computable function $f$ is \emph{quantum total}
%     (q-total) if for each $\mathbf{x}$, $f(\mathbf{x})\in\PP_1$.
%   \item We denote by $\qcf$ the \emph{class of quantum computable
%       functions}.
%   \end{enumerate}
% \end{definition}

\begin{definition}[partial quantum computable functions]\label{definition:qcf}
  \mbox{}
  \begin{enumerate}
  \item A function $f:\ell^2_1 \to \PP$ is \emph{partial quantum
      computable} (q-computable) if there exists a QTM $M$
    s.t. $f(\mathbf{x})=\pd{}$ iff $M_{\nu(\mathbf{x})}\to \pd{}$.
  \item A q-partial computable function $f$ is \emph{quantum total}
    (q-total) if for each $\mathbf{x}$, $f(\mathbf{x})\in\PP_1$.
  \end{enumerate}
  $\qcf$ is the \emph{class of partial quantum computable functions}.
\end{definition}

% \fbox{\blue{A che serve questa def? Mai usata}}

% \fbox{\blue{Poi, siccome le funzioni sono totali, tutto si riduce a =}}
% \begin{definition}[equivalences]
% \begin{enumerate}
% \item[]
% \item
% Two QTMs $M'$ and $M''$ are \emph{input-output equivalent} (notation:
% $\eeq{M'}{M''}$) if for each initial q-configuration $\ket{\phi}$
%  $M'_{\ket{\phi}}\uparrow$ and $M''_{\ket{\phi}}\uparrow$, or
%   there is $\pd{}\in\PP$ s.t. $M'_{\ket{\phi}}\to \pd{}$ and
%   $M''_{\ket{\phi}}\to \pd{}$.
% \item
% For $f,g$ q-computable functions, we define
% $f(\mathbf{x}) \simeq g(\mathbf{x})$ iff $f(\mathbf{x})\uparrow$ and
% $g(\mathbf{x})\uparrow$, or $f(\mathbf{x})= g(\mathbf{x})$.
% \item $f\equiv g$ iff $\forall \mathbf{x}, f(\mathbf{x})\simeq g(\mathbf{x})$.
% \end{enumerate}
% \end{definition}

\section{Observables}\label{Sect:observables}

While the evolution of a closed quantum system (e.g., a QTM) is
reversible and deterministic once its evolution operator is known, a
(global) measurement of a q-configuration is an irreversible process,
which causes the collapses of the quantum state to a new state with a
certain probability. Technically, a measurement corresponds to a
projection on a subspace of the Hilbert space of quantum states. For
the sake of simplicity, in the case of QTMs, let us restrict to
measurements observing if a configuration belongs to the subspace
described by some set of configurations $\B$. The effect of
such a measurement is summarised by the following:

\medskip

\begin{quote}
  \noindent\textbf{Measurement postulate}\\
  Given a set of configurations $\B\subseteq\GConf$, a
  measurement observing if a quantum configuration
  $\ket{\phi}=\sum_{C\in\GConf} e_C \ket{C}$ belongs to the subspace
  generated by $\cb{\B}$ gives a positive answer with a
  probability $p=\sum_{C\in\B}|e_C|^2$, equal to the square
  of the norm of the projection of $\ket{\phi}$ onto
  $\ell^2(\B)$, causing at the same time a collapse of the
  configuration into the normalised projection
  $\sum_{C\in\B} p^{-1}e_C \ket{C}$; dually, it gives a negative answer
  with  probability $1-p=\sum_{C\not\in\B}|e_C|^2$ and a
  collapse onto the subspace $\ell^2(\GConf\setminus\B)$
  orthonormal to $\ell^2(\B)$, that is, into the normalised
  configuration $\sum_{C\not\in\B} (1-p)^{-1}e_C \ket{C}$.
\end{quote}

\medskip

% \noindent\textbf{Measurement postulate}\\
% \textit{Let $\ket{\phi} =\sum e_i \ket{C_i}$ be a q-configuration. For
%   any $i\in\NN$, the \emph{measurement} of $\ket{\phi}$ causes a
%   probabilistic quantum collapse to $\ket{C_i}$ with probability
%   $p=|e_i|^2$ (notation:\ $\qrid{\ket{\phi}}{p}{\ket{C_i}}$).}

% \medskip

% A direct consequence of the measurement postulate and of the
% definition of PPDs is the following:
% \begin{fact}
%   The measurement of $\ket{\phi}$ produces a configuration coding the
%   natural number $n$ with probability $\prob{\ket{\phi}}{n}$.
% \end{fact}

Because of the \textit{irreversible} modification produced by any
measurement on the current configuration, and therefore on the rest of
the computation, we must  deal with the problem of how to read the
result of a computation. In other words, we need to establish some
protocol to observe when a QTM has eventually reached a final
configuration, and to read the corresponding  result.

\subsection{The approach of Bernstein and Vazirani}

We already discussed how \BeV's ``sensible'' QTMs are machines where
all the computations in superposition are in some sense terminating,
and reach the final state at the same time (are ``stationary'', in
their terminology).  More precisely, Definition~3.11 of~\BeV reads: \textit{"A final configuration of a QTM is any configuration in
  [final] state. If when QTM $M$ is run with input $x$, at time $T$
  the superposition contains only final configurations, and at any
  time less than $T$ the superposition contains no final
  configuration, then $M$ halts with running time $T$ on input $x$."}

This is a good definition for a theory of computational complexity
(where the problems are classic, and the inputs of QTMs are always
classic) but it is of little use for developing a theory of effective
quantum functions. Indeed, inputs of a B\&V-QTM \textit{must} be
classical---we cannot extend by linearity a B\&V-QTM on inputs in
$\ell_1^2$, since there is no guarantee whatsoever that on different
inputs the same QTM halts with the same running time.

\subsection{The approach of Deutsch}

Deutsch~\cite{Deu85} assumes that QTMs are enriched with a termination
bit $T$.  At the beginning of a computation, $T$ is set to $0$, and
the machine sets this termination bit to $1$ when it enters into a
final configuration. If we write $\ket{T=i}$ for the function that
returns $1$ when the termination bit is set to $i$, and $0$ otherwise,
a generic q-configuration of a Deutsch's QTM can be written as

\[\ket{\phi}=
  \ket{T=0}\otimes \sum_{C\not\in \FConf} e_C\ket{C} 
  +\ket{T=1}\otimes \sum_{D\in\FConf}d_D\ket{D}
\]

The observer periodically measures $T$ in a non destructive way (that
is, without modifying the rest of the state of the machine).
\begin{enumerate}
\item If the result of the measurement of $T$ gives the value $0$,
  $\ket{\phi}$ collapses (with a probability equal to
  $\sum_{C\not\in\FConf} |e_C|^2$) to the q-configuration
  $$\ket{\psi'}=\frac{\ket{T=0}\otimes 
    \sum_{C\not\in\FConf}
    e_C\ket{C}}{\sum_{C\not\in\FConf}|e_C|^2}$$
  and the computation continues with $\ket{\psi'}$.
\item If the result of the measurement of $T$ gives the value $1$,
  $\ket{\phi}$ collapses (with probability
  $\sum_{D\in\FConf} |d_D|^2$) to
  $$\ket{\psi''}= \frac{\ket{T=1}\otimes
    \sum_{D\in\FConf} d_D\ket{D}}{\sum_{D\in\FConf}
    |d_D|^2}$$
  and, immediately after the collapse, the observer makes a further
  measurement of the component
  $\dfrac{\sum_{D\in\FConf}d_D\ket{D}}{\sum_{D\in\FConf}|d_D|^2}$
  in order to read-back a final configuration.
\end{enumerate}

Note that  Deutsch's protocol (in a irreversible way) spoils at
each step the superposition of configurations.
The main point of Deutsch's approach is that a measurement must be
performed immediately after some computation enters into a final state. In
fact, since at the following step the evolution might lead the machine
to exit the final state modifying the content of the tape, we would
not be able to measure at all this output. In other words, either the
termination bit acts as a trigger that forces a measurement each time
it is set, or we perform a measurement after each step of the
computation.

\subsection{Our approach}

The measurement of the output computed by our QTMs can be performed
following a variant of Deutsch's approach. Because of the
particular structure of the transition function of our QTMs, we shall
see that we do not need  any additional termination bit, that a measurement
can be performed at any moment of the computation, and that indeed we
can perform several measurements at distinct points of the computation
without altering the result (in terms of the probabilistic distribution of
the observed output).

Given a q-configuration $\ket{\phi}=\ket{\phi_f}+\ket{\phi_{nf}}$,
where $\ket{\phi_f}\in\ell^2(\FConf)$ and
$\ket{\phi_{nf}}\in\ell^2(\GConf\setminus\FConf)$, our
\emph{output measurement} tries to get an output value from
$\ket{\phi}$ by the following procedure:
\begin{enumerate}
\item first of all, we observe the final states of $\ket{\phi}$,
  forcing the q-configuration to collapse either into the final
  q-configuration $\ket{\phi_f}/\norm{\ket{\phi_f}}$, or into the
  q-configuration $\ket{\phi_{nf}}/\norm{\ket{\phi_{nf}}}$, which does not
  contain any final configuration;
\item then, if the q-configuration collapses into
  $\ket{\phi_f}/\norm{\ket{\phi_f}}$, we observe one of these configurations,
  say $\ket{C}$, which gives us the observed output $\val[C]=n$,
  forcing the q-configuration to collapse into the final base q-configuration
  $(e_c/|e_c|)\ket{C}$;
\item otherwise, we leave unchanged the q-configuration
  $\ket{\phi_{nf}}/\norm{\ket{\phi_{nf}}}$ obtained after the first
  observation, and we say that we have observed the special value
  $\bot$.
\end{enumerate}

Summing up, an output measurement of $\ket{\phi}$ may lead to observe
an output value $n\in\NN$ associated to a collapse into a base final
configuration $\ket{C}\in\ket{\phi}$ s.t.\ $\val[\phi]=n$ or to
observe the special value $\bot$ associated to a collapse into
a q-configuration which does not contain any final configuration.

\begin{definition}[output observation]\label{def:out-obs}
  An \emph{output observation} with collapsed q-configuration
  $\ket{\psi}$ and \emph{observed output} $x\in\NN\cup\{\bot \}$ is
  the result of an output measurement of the q-configuration
  $\ket{\phi}=\sum_{C\in\GConf}e_C\ket{C}$. Therefore, it is a triple
  $\outobs{\phi}{x}{\psi}$ s.t.\
  \begin{enumerate}
  \item either $x=n\in\NN$, and
    $$\ket{\psi}=\frac{e_c}{|e_c|}\ket{C}
    \qquad\mbox{ with }\qquad
    \ket{C}\in\ket{\phi_{f}} 
    \mbox{ and }
    \val[C]=n$$
  \item or $x=\bot$, and
    $$\ket{\psi}=\frac{\ket{\phi_{nf}}}{\norm{\ket{\phi_{nf}}}}
    \qquad\mbox{ where }\qquad
    \ket{\phi_{nf}}=\sum_{C\not\in\FConf}e_C\ket{C}
    \mbox{ and }
    \norm{\ket{\phi_{nf}}} \neq 0$$
  \end{enumerate}
  The \emph{probability of an output observation} is defined by
  $$\pr{\outobs{\phi}{x}{\psi}} =
  \begin{cases}
    |e_C|^2 & \qquad \mbox{if } x=n\in\NN \\[+1ex]
    \norm{\ket{\phi_{nf}}}^2 & \qquad \mbox{if } x=\bot
  \end{cases}$$
\end{definition}

\begin{remark}\label{rem:obs-base-final}
  Let $e\outobs{C}{x}{\phi}$, with $C\in\FConf$ and
  $\val[C]=n$. By definition, $x=n$ and $\ket{\phi}=(e/|e|)\ket{C}$;
  moreover, $\pr{e\outobs{C}{x}{\phi}}=|e|^2$.
\end{remark}

\begin{remark}\label{rem:ortonorm-obs}
  For every distinct pair of output observations
  $\outobs{\phi}{x_1}{\psi_1}$ and $\outobs{\phi}{x_2}{\psi_2}$, we
  have that $\psi_1$ and $\psi_2$ are in the orthonormal subspaces generated by
  the two disjoint sets $\B_1,\B_2\subseteq\GConf$, where
  $\B_i=\{C\in\GConf\mid \ket{C}\in\ket{\psi_i}\}$.
\end{remark}

\begin{definition}[observed run]
  Let $M$ be a QTM and $U_M$ its time evolution operator. For any
  monotone increasing function $\tau:\NN\to\NN$ (that is,
  $\tau(i)<\tau(j)$ for $i<j$):
  \begin{enumerate}
  \item a \emph{$\tau$-observed run} of $M$ on the initial
    q-configuration $\ket{\phi}$ is a sequence
    $\{\ket{\phi_i}\}_{i \in \NN}$ s.t.:
    \begin{enumerate}
    \item $\ket{\phi_0} = \ket{\phi}$;
    \item $U_M\outobs{\phi_{h}}{x_i}{\phi_{h+1}}$, when $h=\tau(i)$ for
      some $i\in\NN$;
    \item $\ket{\phi_{h+1}}=U_M\ket{\phi_{h}}$ otherwise.
    \end{enumerate}
  \item A \emph{finite $\tau$-observed run} of length $k$ is any
    finite prefix of length $k$ of some $\tau$-observed run. Notation:
    if $R=\{\ket{\phi_i}\}_{i \in \NN}$, then
    $R[k] = \{\ket{\phi_i}\}_{i \leq k}$.
  \end{enumerate}
\end{definition}

\begin{remark}\label{rem:obs-out-run}
  We stress that, given an $R=\{\ket{\phi_i}\}_{i \in \NN}:$
  \begin{enumerate}
  \item either it never obtains a value $n\in\NN$ as the result of an
    output observation, and then it never reaches a final
    configuration;
  \item or it eventually obtains such a value collapsing the
    q-configuration into a base final configuration $u\ket{C}$ s.t.\
    $|u|=1$ and $\val[C]=n$, and from that point onward all the
    configurations of the run are base final configurations
    $u\ket{C_j}=u\,U^j\ket{C}$ s.t.\ $\val[C_j]=n$, and all the
    following observed outputs are equal to $n$ (see
    Remark~\ref{rem:obs-base-final}).
  \end{enumerate}
\end{remark}
\begin{definition}
  Let $R=\{\ket{\phi_i}\}_{i \in \NN}$ be a $\tau$-observed run.
  \begin{enumerate}
  \item The sequence $\{x_i\}_{i\in\NN}$ s.t.\
    $\outobs{\phi_{h}}{x_i}{\phi_{h+1}}$, with $h=\tau(i)$, is the
    \emph{output sequence} of the $\tau$-observed run $R$.
  \item The \emph{observed output} of $R$ is the value
    $x\in\NN\cup\{\bot\}$ (notation: $\obsout{R}{x}$) defined by:
    \begin{enumerate}
    \item $x=n\in\NN$, if $x_i=n$ for some $i\in\NN$;
    \item $x=\bot$ otherwise.
    \end{enumerate}
  \item For any $k$, the output sequence of the finite $\tau$-observed
    run $R[\tau(k)]$, is the finite sequence $\{x_i\}_{i\leq\tau(k)}$
    and $x_k$ is its observed output.
  \end{enumerate}
\end{definition}

\begin{definition}[probability of a run]
  Let $R=\{\ket{\phi_i}\}_{i \in \NN}$ be a $\tau$-observed run.
  \begin{enumerate}
  \item For $k\in\NN$, the probability of the finite $\tau$-observed
    run $R[k]$ is inductively defined by
    \begin{enumerate}
    \item $\pr{R[0]} = 1$;
    \item $\pr{R[k+1]} = 
      \begin{cases}
        \pr{R[k]}\,\pr{\outobs{\phi_k}{x_i}{\phi_{k+1}}}
        & \
        \parbox{16ex}{%
          when $k = \tau(i)$ for some $i\in\NN$ 
        }
        \\[+1.5em]
        \pr{R[k]} & \ \mbox{otherwise }
      \end{cases}$
    \end{enumerate}
  \item $\pr{R} = \lim_{k\to\infty}\pr{R[k]}$. 
  \end{enumerate}
\end{definition}

We remark that $\pr{R}$ is well-defined, since
$1\geq \pr{R[i]} \geq \pr{R[j]} > 0$, for every $i\leq j$. Therefore,
$$1\geq \pr{R} = \lim_{k\to\infty}\pr{R[k]}=\inf\{\pr{R[k]}\}_{k\in\NN}\geq 0.$$

\begin{remark}\label{rem:pr-runs}
  Let $R=\{\phi_i\}_{i\in\NN}$ be a $\tau$-observed run s.t.\
  $\obsout{R}{n}$, for some $n\in\NN$. As observed in
  Remark~\ref{rem:obs-out-run}, for some $k$, we have
  $\obsout{R[\tau(k)]}{\bot}$ and $\obsout{R[\tau(k)+1]}{n}$;
  moreover, for $i > k$, $\ket{\phi_i}=u\ket{C_i}$ with $|u|=1$, 
  $C_i\in\FConf$, and $\val[C_i]=n$. As a consequence,
  $\pr{R[k+1]}=\pr{R[i]}=\pr{R}$, for $i>k$ (since, by
  Remark~\ref{rem:obs-base-final},
  $\pr{\outobs{\phi_{\tau(i)}}{n}{\phi_{\tau(i)+1}}}=1$).
\end{remark}

\begin{definition}[observed computation]\label{def:obs-computation}
  The \emph{$\tau$-observed computation} of a QTM $M$ on the initial
  q-configuration $\ket{\phi}$, is the set
  $\K_{\ket{\phi},\tau}^M$ of the $\tau$-observed runs of $M$
  on $\ket{\phi}$ with the measure
  $\PR:\varP(\K_{\ket{\phi},\tau}^M)\to\CC$ defined by
  $$\Pr\B = \sum_{R\in\B}\pr{R}$$
  for every $\B\subseteq\K_{\ket{\phi},\tau}^M$.
\end{definition}

By $\K[k]_{\ket{\phi},\tau}^M$ we shall denote the
set of the finite $\tau$-observed runs of length $k$ of $M$ on
$\ket{\phi}$, with the measure $\PR$ on its subsets (see
Definition~\ref{def:obs-computation}).

It is immediate to observe that the set $\K_{\ket{\phi},\tau}^M$ naturally defines an infinite tree labelled with q-configurations where each infinite path starting from the root $\ket{\phi}$ correspond to a $\tau$-observed run in $\K_{\ket{\phi},\tau}^M$.

\begin{lemma}\label{lem:ortonorm-runs}
  Given $R_1,R_2\in\K_{\ket{\phi},\tau}^M$, with
  $R_1=\{\phi_{1,i}\}_{i\in\NN}\neq \{\phi_{2,i}\}_{i\in\NN}=R_2$,
  there is $k\geq 0$ s.t.\
  \begin{enumerate}
  \item $\phi_{1,i}=\phi_{2,i}$ for $i\leq \tau(k)$, that is,
    $R_1[\tau(k)]=R_2[\tau(k)]$;
  \item for $i>\tau(k)$, the q-configurations
    $\phi_{1,i}\neq \phi_{2,i}$ are in two orthonormal subspaces
    generated by two distinct subsets of $\GConf$.
  \end{enumerate}
\end{lemma}
\begin{proof}
  Let $R_1[h]=R_2[h]$ be the longest common prefix of $R_1$ and
  $R_2$. Since they both starts with $\ket{\phi}$, such prefix is not
  empty; moreover, by the definition of $\tau$-observed run, it is
  readily seen that $h=\tau(k)$, for some $k$. By construction,
  $\phi_{1,h+1}\neq \phi_{2,h+1}$ and
  $\obsout{\phi_h}{x_j}{\phi_{j,h+1}}$, for $j=1,2$, with
  $\phi_h=\phi_{1,h}=\phi_{2,h}$. Moreover, at least
  one of the two q-configurations
  $\ket{\psi_{1,h+1}},\ket{\psi_{2,h+1}}$ is a final base
  q-configuration $u\ket{C_1}$; for instance, let
  $\ket{\psi_{1,h+1}}=u_1\ket{C_1}$.

  Let us take
  $\B_{a,i}=\{C\in\GConf\mid \ket{C}\in\ket{\psi_{a,i}}\}$,
  for $a=1,2$ and $i>h$. We prove that
  $\B_{1,i} \cap\B_{2,i}=\emptyset$, for $i>h$.
  First of all, this holds for $i=h+1$, by
  Remark~\ref{rem:ortonorm-obs}.  Then, we distinguish two cases:
  \begin{enumerate}
  \item $\ket{\psi_{2,h+1}}=u_2\ket{C_2}$ is a final base
    configuration. We have that
    $\ket{\psi_{1,i+1}}=u_1\,U^{i-h}\ket{C_1} \neq
    u_2\,U^{i-h}\ket{C_2}=\ket{\psi_{2,i+1}}$,
    for $a=1,2$ and $i\geq h$ (by Remark~\ref{rem:obs-out-run} and the
    fact that $U_M$ is injective, since we are already remarked that
    $C_1\neq C_2$).
  \item $\ket{\psi_{2,h+1}}\in\ell^2(\GConf\setminus \FConf)$.  Every
    $\ket{C}\in \ket{\psi_{2,h+1}}$ contains less than $i-h$ extra
    symbols, while $\ket{\psi_{1,i+1}}=u_1\,U^{i-h}\ket{C_1}$ contains
    at least $i-h$ extra symbols.  Therefore,
    $\B_{1,i+1}\cap\B_{2,i+1} =\{U^{i-h}\ket{C_1}\}\cap\B_{2,i+1}
    =\emptyset$.
  \end{enumerate}
\end{proof}

% \begin{lemma}
%   For any initial q-configuration $\ket{\phi}$ and every $\tau$, we have that 
%   \begin{enumerate}
%   \item $\PR\, \K[k]_{\ket{\phi},\tau}=1$, for every $k\in\NN$;
%   \item $\PR\, \K_{\ket{\phi},\tau}=1$.
%   \end{enumerate}
% \end{lemma}
% \begin{proof}
%   ...
% \end{proof}

\begin{lemma}\label{lem:obs-reconstr-qconf}
  Let $K _{\ket{\phi}}^M=\{\phi_i\}_{i\in\NN}$ be the computation
  of the QTM $M$ on the initial q-configuration $\ket{\phi}$ and
  $\K_{\ket{\phi},\tau}^M$ the $\tau$-observed computation on
  the same initial configuration.  For every $k\in\NN$, we have
  that
  $$\ket{\phi_k} = \sum_{R\in\K[k]_{\ket{\phi},\tau}^M} \pr{R}\,\ket{\psi_{R}}$$
  where $\ket{\psi_R}$ is the last q-configuration of the finite run
  $R$ of length $k$.
\end{lemma}
\begin{proof}
  By definition, $\phi=\phi_0$ and $R=\{\phi\}$ with $\pr{R}=1$ and
  $\psi_R=\phi$, is the only run of length $0$ in
  $\K_{\ket{\phi},\tau}^M$. Therefore, the assertion
  trivially holds for $k=0$. 

  Let us then prove the assertion by induction on $k$. By definition
  and the induction hypothesis
  $$\ket{\phi_{k+1}} = U_M\ket{\phi_k} = 
  \sum_{R\in\K[k]_{\ket{\phi},\tau}^M} \sqrt{\pr{R}}\,U_M\ket{\psi_R}$$
  We have two possibilities:
  \begin{enumerate}
  \item $k\neq\tau(i)$ for any $i$. In this case, there is a bijection
    between the runs of length $k$ and those of length $k+1$, since
    each run $R'\in\K[k+1]_{\ket{\phi},\tau}^M$ is obtained
    from a path $R\in\K[k]_{\ket{\phi},\tau}^M$ with last
    q-configuration $\ket{\psi_R}$, by appending to $R$ the
    q-configuration $\ket{\psi_{R'}}=U_M\ket{\psi_R}$. Moreover, since
    by definition, $\pr{R'}=\pr{R}$, we can conclude  that
    $$\ket{\phi_{k+1}} = 
    \sum_{R\in\K[k]_{\ket{\phi},\tau}^M} \pr{R}\,U_M\ket{\psi_R} =
    \sum_{R'\in\K[k+1]_{\ket{\phi},\tau}^M} \pr{R'}\ket{\psi_{R'}}$$
  \item $k=\tau(i)$, for some $i$. In this case, every
    $R\in\K[k]_{\ket{\phi},\tau}^M$ with last q-configuration
    $\ket{\psi_R}$ generates a run $R'$ of length $k+1$ for every
    output observation $\outobs{\psi_R}{x}{\psi_{R'}}$, where $R'$ is
    obtained by appending $\ket{\psi_{R'}}$ to $R$. Therefore, let
    $R=\{\ket{\psi_i}\}_{i\leq k} \in
    \K[k]_{\ket{\phi},\tau}^M$ and
    $$\B_R = \{ \{\psi_i\}_{i\leq k+1} \mid
    \outobs{\psi_k}{x}{\psi_{k+1}}\}$$
    by applying Definition~\ref{def:out-obs}, we easily check that
    $$U_M\ket{\psi_R} = \sum_{R'\in\B_R}\pr{\outobs{\psi_R}{x}{\psi_{R'}}}\ket{\psi_{R'}}$$
    Thus, by substitution, and
    $\pr{R}\,\pr{\outobs{\psi_R}{x}{\psi_{R'}}}=\pr{R'}$
    \begin{align*}
      \ket{\phi_{k+1}} 
      & = 
        \sum_{R\in\K[k]_{\ket{\phi},\tau}^M} \pr{R}\,U_M\ket{\psi_R} 
      \\
      & =
        \sum_{R\in\K[k]_{\ket{\phi},\tau}^M}
        \sum_{R'\in\B_R}
        \pr{R}\,\pr{\outobs{\psi_R}{x}{\psi_{R'}}}\ket{\psi_{R'}}
      \\
      & = 
        \sum_{R'\in\,\bigcup_{R\in\K[k]_{\ket{\phi},\tau}^M}\B_R} \pr{R'}\ket{\psi_{R'}}
      \\
      & = 
        \sum_{R'\in\K[k+1]_{\ket{\phi},\tau}^M} \pr{R'}\ket{\psi_{R'}}
    \end{align*}
    since
    $\bigcup_{R\in\K[k]_{\ket{\phi},\tau}^M}\B_R=\K[k+1]_{\ket{\phi},\tau}^M$.
  \end{enumerate}

\end{proof}

We are finally in the position to prove that our observation protocol
is compatible with the probability distributions that we defined as
computed output of a QTM computation.

\begin{theorem}
  Let $K _{\ket{\phi}}^M=\{\phi_i\}_{i\in\NN}$ be the computation
  of the QTM $M$ on the initial q-configuration $\ket{\phi}$ and
  $\K_{\ket{\phi},\tau}^M$ the $\tau$-observed computation on
  the same initial configuration. For every $n\in\NN$:
  \begin{enumerate}
  \item
    $\pd{\ket{\phi_k}}(n) = \pr{R\in
      \K[k]_{\ket{\phi},\tau}^M \mid \obsout{R}{n}}$,
    for every $k=\tau(i)$, with $i\in\NN$;
  \item $\pd{K _{\ket{\phi}}^M}(n) = \pr{R\in \K_{\ket{\phi},\tau}^M \mid \obsout{R}{n}}$.
  \end{enumerate}
\end{theorem}
\begin{proof}
  By Lemma~\ref{lem:obs-reconstr-qconf}, we know that
  $\ket{\phi_k} = \sum_{R\in\K[k]_{\ket{\phi},\tau}^M}
  \pr{R}\,\ket{\psi_{R}}$,
  where $\ket{\psi_R}$ is the last q-configuration of $R$. Since
  $k=\tau(i)$, for some $i$, we also know that either
  $\psi_R\in\GConf\setminus\FConf$ or $\ket{\psi_R}=u_R\ket{C_R}$
  with $C_R\in\FConf$ and $|u_R|=1$. Therefore,
  $$\pd{\ket{\phi_k}}(n)=
  \norm{\sum_{R\in\B[k,n]} \sqrt{\pr{R}}\,u_R\ket{C_{R}}}^2$$
  where
  \begin{align*}
    \B[k,n] 
    & = 
      \{R\in \K[k]_{\ket{\phi},\tau}^M \mid \obsout{R}{n}\}
    \\
    & = 
      \{R\in\K[k]_{\ket{\phi},\tau}^M \mid
      \ket{\psi_R}=u_R\ket{C_R} \mbox{ with } \val[C_R]=n\}
  \end{align*}
  
  By Lemma~\ref{lem:ortonorm-runs}, we know that for every
  $R_1,R_2\in\B[k,n]$, we have
  $\ket{C_{R_1}}\neq\ket{C_{R_2}}$. Therefore
  $$\pd{\ket{\phi_k}}(n)=
  \sum_{R\in\B[k,n]} \pr{R}|\,u_R|^2 =
  \sum_{R\in\B[k,n]} \pr{R} =\ \PR \B[k,n]$$
  since $|u_R|=1$. Which concludes the proof of the first item of the
  assertion.

  In order to prove the second item, let
  $\B[\omega,n] = \{R\in \K_{\ket{\phi},\tau}^M \mid
  \obsout{R}{n}\}$. We have that
  \begin{multline*}
    \pr{R\in \K_{\ket{\phi},\tau}^M \mid \obsout{R}{n}}
    = \sum_{R\in \B[\omega,n]}\pr{R}
    = \lim_{k\to\infty} \sum_{\dind{R\in\B[\omega,n]}{\obsout{R[\tau(k)]}{n}}} \pr{R}
    \\ 
    {}^{(*)}= \lim_{k\to\infty} 
        \sum_{\dind{R\in\B[\omega,n]}{\obsout{R[\tau(k)]}{n}}} \pr{R[\tau(k)]}
    \\
    {}^{(**)}= \lim_{k\to\infty}\sum_{R'\in \B[\tau(k),n]}\pr{R'}
    = \lim_{k\to\infty} \PR\, \B[\tau(k),n]
  \end{multline*}
  since: $(*)$ $\pr{R}=\pr{R[\tau(k)]}$, when $\obsout{R[\tau(k)]}{n}$ (see
  Remark~\ref{rem:pr-runs}); $(**)$ there is a bijection between the
  sets $S[\tau(k),n]$ and
  $\{R\in \B[\omega,n] \mid \obsout{R[\tau(k)]}{n}\}$ (see
  Remark~\ref{rem:obs-out-run}) mapping every
  $R'\in\B[\tau(k),n]$ with last q-configuration $u\ket{C}$
  into
  $R = \{\psi_{R,i}\}_{i\in\NN}\in \{R\in \B[\omega,n]
  \mid \obsout{R[\tau(k)]}{n}\}$
  s.t\ $R'=R[\tau(k)]$ and
  $\ket{\psi_{R,i}}=u\,U_M^{i-\tau(k)}\ket{C}$ for $i\geq \tau(k)$.
  
  Therefore, by the (already proved) first item of the assertion
  $$\pr{R\in \K_{\ket{\phi},\tau}^M \mid \obsout{R}{n}}
  = \lim_{k\to\infty}\PR\, \B[\tau(k),n]
  = \lim_{k\to\infty}\pd{\ket{\phi_{\tau(k)}}}(n)
  = \pd{K _{\ket{\phi}}^M}(n)$$
\end{proof}

\section{Remarks on the expressive power}\label{Sec:ExpressivePower}

%\fcolbox{%
%    Ho portato qui la parte sulle configurazioni computabili.
%    
%    Dividerei questa sezione in due parti:
%
%    - una che parla dei reali calcolabili, non so bene il titolo, ma
%    da notare que noi abbiamo già usato ``computable'' per le
%    configurazioni computabili, quindi specificherei che qui stiamo comparando
%    con la calcolabilità classica
%
%    - seconda parte di comparazione con \BeV e gli esempi
%
%    Comunque eviterei la notazione dei calcolabili. Io ho spostato il
%    commento sui reali calcolabili all'inizio. Senza appesantire
%    ulteriormente la notazione, si potrebbe dire che adesso
%    consideriamo il caso in cui tutti i reali sono calcolabili.
%}

\subsection{Computable configurations}

Since QTMs represent (ideal) physically realisable devices, we should
constrain the complex numbers in the time evolution operator to be
computable (see, e.g., Remark 9.2 in \cite{kitaev} %pag.\ 90
for a discussion).

\begin{definition}[computable numbers]\label{Def-computable-real}
  A real number $x$ is computable if there exists a
  deterministic Turing machine that on input $\nstring{n}$ computes a binary
  representation of an integer $m\in\ZZ$ such that $|\frac{m}{2^n} -
  x|\leq \frac{1}{2^n}$.

  The computable complex numbers $\CoC$ are those complexes whose real
  and imaginary parts are both computable.
\end{definition}

\begin{definition}[computable QTM]\label{Def-computable-qtm}
  A QTM is \emph{computable} iff 
  for any $(q,a)\in \varS_0$ and every $(p,b,d)\in \T_0\times\DD$, we have $\delta_0(q,a)(p,b,d)\in\CoC$.
\end{definition}

Observe that, being $(\Q_0\cup \Q_t) \times \Sigma \times \DD$ finite, in a computable QTM
the element $\delta_0(q,a)\in \ell^2((\Q_0\cup \Q_t) \times \Sigma \times \DD)$ may be ``effectively presented'' in an obvious way.

We postpone to a subsequent paper a full treatment of computable QTMs,
and especially of the computability theory they may engender.
We make here only some simple, preliminary remarks. 

%\fcolbox[yellow]{ \red{\textbf{Simone dice: troppo debole, da togliere}}\\
%\blue{\textbf{Andrea: togliere e sostituire con qualcosa più forte?}}
%\begin{definition}[computable configurations]\label{Def-computable-config}
%  A q-configuration $\ket{\phi}$ is \emph{computable} iff
%  $\phi(\GConf_{\Sigma,\Q})\subseteq\CoC$.
%\end{definition}
%}

First, computable QTMs form a recursive enumerable class. Indeed, 
any complex number $e\in\CoC$ may be described by the
index of the (classical) TM computing it (write $\llceil e\rrceil\in\NN$ for this).
Moreover, for any $(q,a)\in \varS_0$ we have a classical TM enumerating  $\delta_0(q,a)$ 
(that is, producing the family of the indexes $\llceil e\rrceil$ of the TMs 
computing the amplitudes).

Second, the time evolution operator $U_M$ of a computable QTM defines
a classically computable function on the (code of) quantum
configurations. We spell this out in case of finite
q-configurations\footnote{The argument may be generalised to infinite
  q-configurations---that is, q-configurations in which there is an
  infinite number of non-zero configurations in
  superpositions---provided these infinite superpositions are
  recursively enumerable.}.

Let us extend the alphabet $\Sigma$ of the QTM $M$ into
$\Sigma_c=\Sigma\cup \Q\cup\{\langle,\rangle\,\star\}$.  Any
finite q-configuration
$\ket{\phi}=\sum_{i=1}^n e_i \ket{\langle\alpha_i,q^i,\beta_i\rangle}$
may be coded by the string
$$\llceil \ket{\phi}\rrceil = \star\nstring{\llceil e_1 \rrceil}\langle\alpha_1 q^1\beta_1\rangle\star\cdots
\star\nstring{\llceil e_n \rrceil}\langle\alpha_n q^n\beta_n\rangle\star\in\Sigma_c^*.$$
Let us denote by $\llceil {\qConf_M}\rrceil$ the set of such codes of finite
q-configurations. 

% $\{\llceil \ket{\phi}\rrceil: \ket{\phi}\in\qConf_M\}$
% commentato perché dovremmo dire che le \qConf_M sono quelle finite; e siamo 
% tutto sommato abbastanza informali sin qui.

The function ${U^c_M}:\llceil {\qConf_M}\rrceil\to
\llceil {\qConf_M}\rrceil $ defined by:
$$
{U^c_M}(\star\nstring{\llceil e_1\rrceil}\langle\alpha_1
q^1\beta_1\rangle\star\cdots \star\nstring{\llceil e_n\rrceil}\langle\alpha_n
q^n\beta_n\rangle\star)=\llceil U_M(\sum_{i=1}^ne_i\ket{\langle\alpha_i,q^i,\beta_i\rangle})\rrceil
$$
is intuitively computable, and therefore, by Church's thesis, is
computable.  Therefore it is only a matter of routine to prove the
following theorem:

\begin{theorem}[classical soundness]
  Let $M$ be a computable QTM such that for each finite input
  $\ket{\phi}$ the corresponding computation is finitary. There is a
  classical partial computable function
  $$\mathsf{Comp}_M:
  \llceil {\qConf_M}\rrceil \partialto \llceil {\qConf_M}\rrceil$$
  s.t.  for each finite $\ket{\phi}\in\cqConf{}$,
  % and for each finite
  % $\rreg{v}\in\cregs$, there exists $\kket{\psi}\in\cqConf{}$
  % s.t.
  % $\eval\;\kket{\psi}=\rreg{v}$ 
  % and  
  $M_{\ket{\phi}}\to \pd{}$ iff there is $\ket{\psi}$
  s.t. $\pd{}=\pd{\ket{\psi}}$ and
  $\mathsf{Comp}_M(\llceil\ket{\phi} \rrceil)
  =\llceil\ket{\psi}\rrceil $.
\end{theorem}

\subsection{A comparison with Bernstein and Vazirani's QTMs: part 2}

In view of Theorem~\ref{theor:fromBeVtoOur}, we may say that our QTMs
generalise B\&V-QTMs, which may be simulated. The general framework,
however, is substantially modified and the ``same'' machine behaves in
different ways in the two approaches. We give two simple examples of
this, before concluding the paper.

\begin{figure}[htb]
  \begin{center}
    \scalebox{0.90}
    {
      \includegraphics{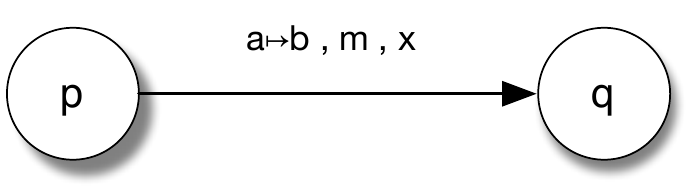} 
    }
  \end{center}
  \caption{The transition function $\delta$}\label{fig:uno}
\end{figure}
 
In this section we shall use a pictorial representation of QTMs, via a
graph for the transition function $\delta$.  When
$\delta(q,a)(p,b,d)=x$ with $x\neq 0$ we draw the labelled arc as in
Figure~\ref{fig:uno}.

\begin{figure}[htb]
  \begin{center}
    \scalebox{0.70}{%
      \includegraphics{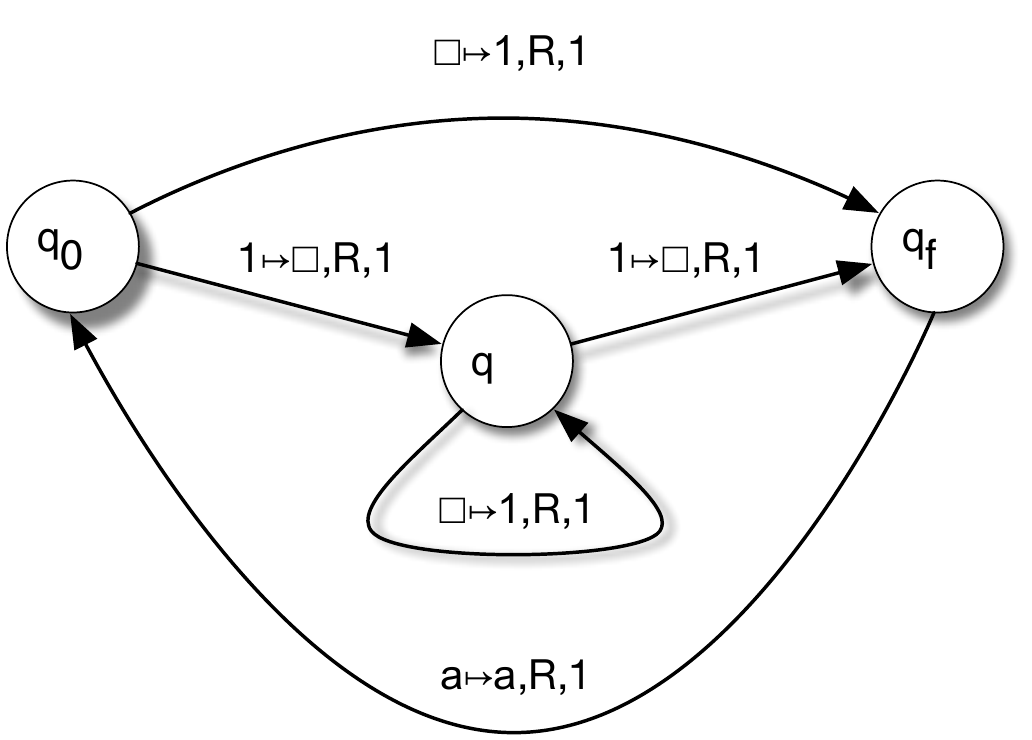} 
    }
  \end{center}
  \caption{Reversible TM a l\`a Bernstein and Vazirani}\label{fig:RTMBeV}
\end{figure}

\begin{example}[\textbf{classical reversible TM with quantum behaviour}]
  Let $M$ be the reversible TM represented in Figure~\ref{fig:RTMBeV},
  where $a\in\{\Box,1\}$. $M$ is a B\&V-QTM indeed.

  If we feed $M$ with a non classical input, e.g.\
  $\ket{\psi}=\frac{1}{\sqrt{2}}\ket{\nstring{1}}+
  \frac{1}{\sqrt{2}}\ket{\nstring{3}}$,
  then $M$ fails to give an answer according to B\&V's framework,
  since B\&V-QTMs always presuppose a classical input.

  \begin{figure}[htb]
    \begin{center}
      \scalebox{0.70}{%
        \includegraphics{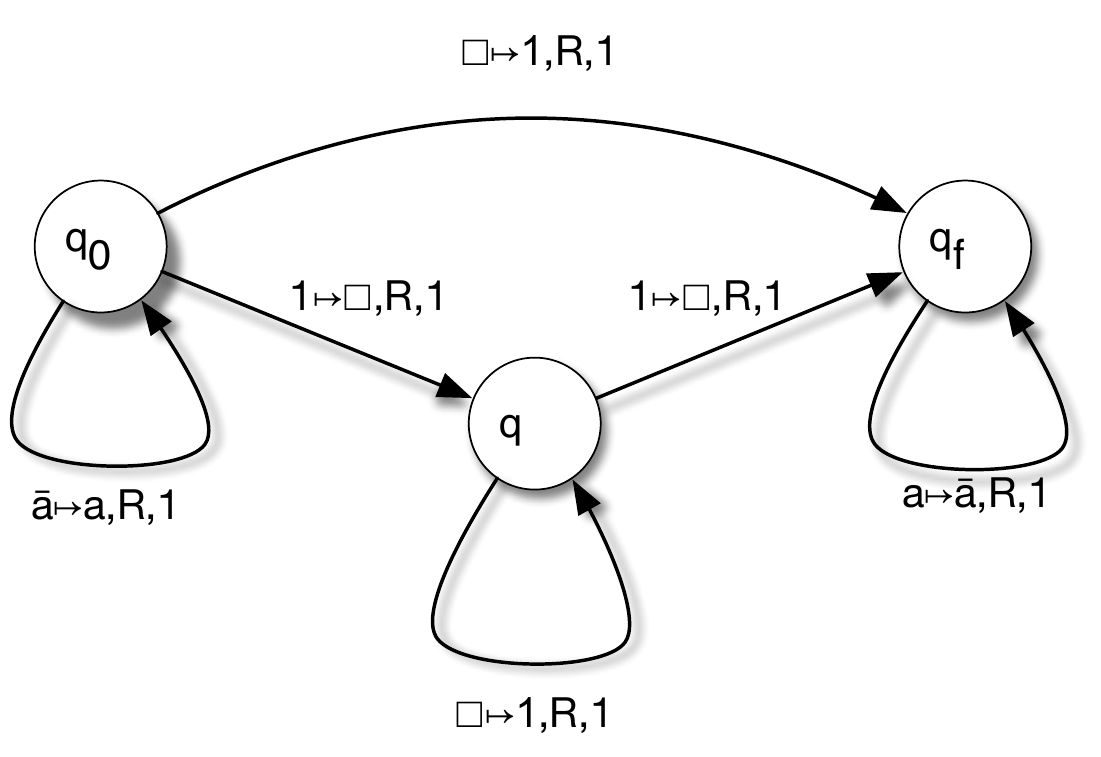} 
      }
    \end{center}
    \caption{Reversible TM}\label{fig:RTM}
  \end{figure}

  If we transform $M$ in our formalism (see
  Theorem~\ref{theor:fromBeVtoOur}), we obtain the QTM in
  Figure~\ref{fig:RTM}.
  From the definition of computed output, we have that
  $M_{\frac{1}{\sqrt{2}}\ket{\nstring{1}}+\frac{1}{\sqrt{2}}\ket{\nstring{3}}}\to
  \{\frac{1}{2}:2\}$;
  namely, with probability $\frac{1}{2}$ the QTM halts with output $2$;
  while with probability $\frac{1}{2}$ it diverges.
\end{example}

\begin{example}[\textbf{A PD obtained as a limit}]
  The following example shows a machine which produces a PD only as an
  infinite limit.  Let us consider the QTM in Figure~\ref{fig:ide},
  where $\Sigma=\{\$,1,\Box\}$, $a\in\{1,\Box\}$, $p$ is a target
  state, and $s$ is a source state.
  A simple calculation show that $M_{\ket{\nstring{n}}}\to\{1: n+1\}$;
  namely, the machine $M$ on input $n$ produces with probability $1$
  the successor $n+1$.  We can see also that the PD $\{1: n+1\}$ is
  obtained only as a limit.  Of course we must not wait an infinite
  time to readback the result!  A correct way to interpret this fact,
  is that for each $n\in\NN$, each $\epsilon\in (0,\frac{1}{2}]$ there
  exist a natural number $j$ s.t.
  $\pd{U^j\ket{\nstring{n}}}(n+1) > 1-\epsilon$.
\end{example}

\begin{figure}[htb]
  \begin{center}
    \scalebox{0.80}{%
      \includegraphics{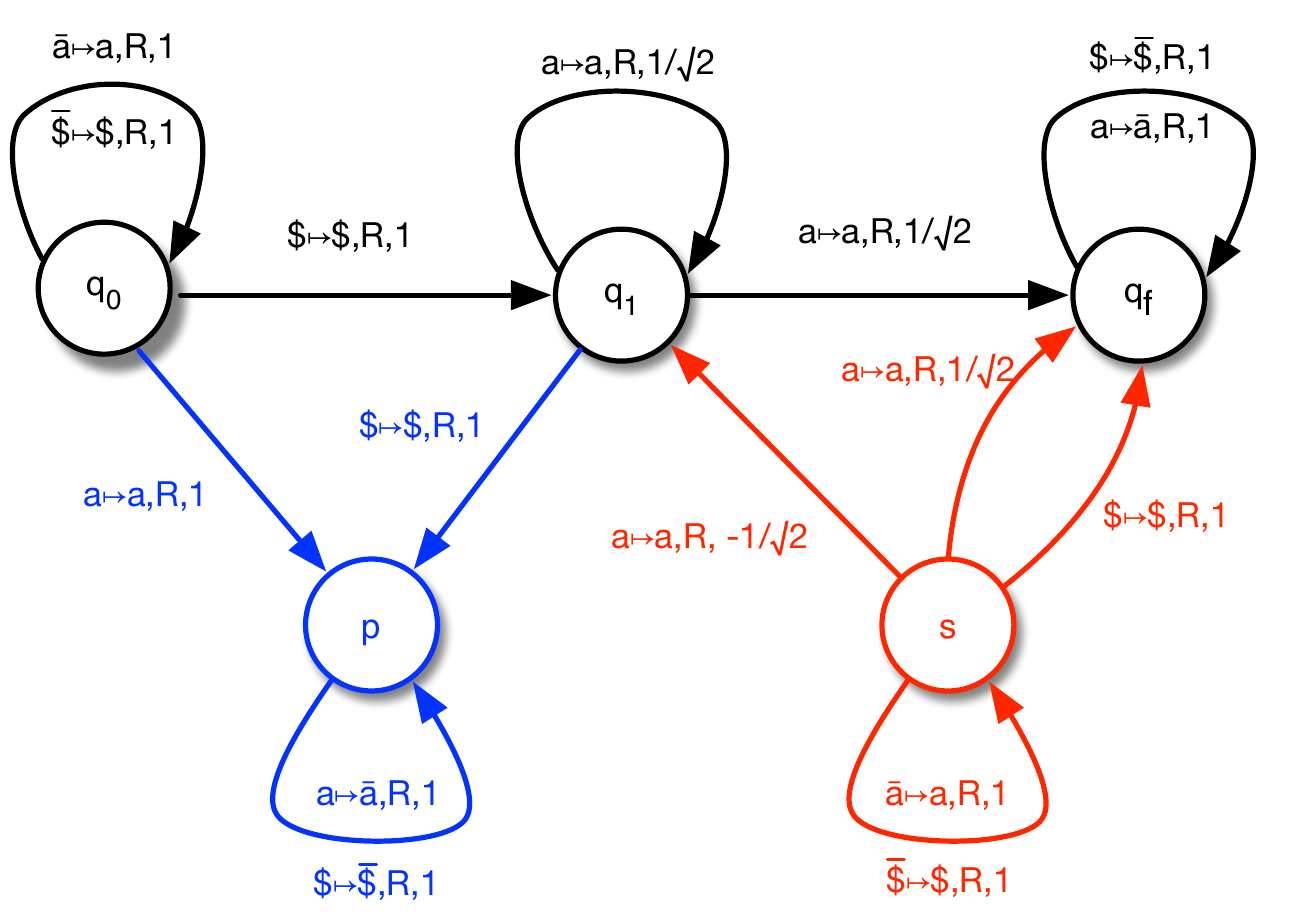} 
    }
  \end{center}
  \caption{Identity}\label{fig:ide}
\end{figure}
  
%  \begin{figure}
%     \begin{center}
%       \scalebox{0.40}
%       {
%         \includegraphics{./FIGURE/Nat.pdf} 
%       }
%     \end{center}
%     \caption{final segments}\label{fig:tre}
%   \end{figure}
   
\section{Conclusions and further work}

We find surprising that in the thirty years since~\cite{Deu85} a theory of quantum computable functions did not develop, and that the main interest remained in QTMs as computing devices for classical problems/functions. This in sharp contrast with the original (Feynman's and Deutsch's) aim to have a better computing simulation of the physical world. 

As always in these foundational studies, we had to go back to the basics,
and look for a notion of QTM general enough to encompass  previous approaches (for instance, simulation of B\&V-QTMs, Theorem~\ref{theor:fromBeVtoOur}), and still sufficiently constrained to allow for a neat mathematical framework (for instance, monotonicity of quantum computations, Theorem~\ref{theor:monot}, a consequence of the particular way final states are treated in order to defuse quantum interference once such states are entered). While several details of the proposed approach may well change during further study, we are particularly happy to have a recursive enumerable class of QTMs. This may allow a fresh look to the problem of a quantum universal machine, and, therefore, to obtain some of the ``standard'' theorems of classical computability theory (s-m-n, normal form, recursion, etc.). These themes, as well those related to the various degrees of partiality of quantum computable functions (see the brief discussion after Proposition~\ref{prop:partiality}) will be the subject of forthcoming papers. 

% \subsection{the mathematical point of view}
% We have characterised quantum computable function as a particular
% subclass of $(\ell_1^2\partialto\mdue)$.
% The choice of $\mdue$ is quite natural, since $\mdue$ 
% associates to each natural number $n$ all the possible quantum  amplitudes
% associated to it by a quantum computation of a given QTM.

% Nevertheless, $\mdue$ is not a standard mathematical space, it raises
% naturally form our definition on QTM.

% From a more ``conservative'' point of view, the question is, what is a
% quantum computable function in $(\ell_1^2\partialto\ell_1^2)$?

% This question leads to the following answer.

% A good question is: what is the relationship between $\qcf$ and
% $\qcfl$? More precisely, when we limit ourselves to the class $\qcfl$,
% do we lose something in expressive power?

% \subsection{the physical point of view}
% \red{Qui ci dovremmo rifare all'idea originale dietro alla nascita del
%   QC. In pressoché tutta la letteratura il QC e' visto solo come un
%   modo per calcolare funzioni classiche. L'idea originale di Feynman
%   era però ben diversa: QC come modo per simulare la fisica/natura.
% Questo da forza allo studio degli spazi delle funzioni quantum-comp.
% }

\bibliographystyle{abbrv}
\bibliography{biblio} 

\appendix
\input{appendix0}

\input{appendix}

\end{document}

%% file: appendix0.tex
% -*- Mode: LaTeX; TeX-master: "bozza-4.0.tex" -*- 

\section{Hilbert spaces with denumerable basis}
\label{sec:HS}

\begin{definition}[Hilbert space of configurations]
  Given a denumerable set $\B$, with $\ell^2(\B)$ we shall denote the
  infinite dimensional Hilbert space defined as follow.

  The set of vectors in $\ell^2(\B)$ is the set
  $$
  \left\{\phi\;|\;\phi:\B\rightarrow\CC, \sum_{C\in
      \B}|\phi(C)|^2 < \infty\right\}
  $$ 
  and equipped with:
  \begin{enumerate}
  \item An {inner sum} $+ : \ell^2(\B) \times\ell^2(\B)\to\ell^2(\B)$
    \\
    defined by $(\phi+\psi)(C)= \phi(C)+\psi(C)$;
  \item A {multiplication by a scalar}\quad
    $\cdot:\CC\times\ell^2(\B)\to \ell^2(\B)$
    \\
    defined by $(a\cdot \phi)(C)= a\cdot(\phi(C))$;
  \item An {inner product}\footnote{%
      The condition $\sum_{C\in\B}|\phi(C)|^2<\infty$ implies that
      $\sum_{C\in\B}\phi(C)^*\psi(C)$ converges for every pair of
      vectors.}
    $\inprod{\cdot}{\cdot}: \ell^2(\B)\times\ell^2(\B)\to\CC$\\
    defined by $\inprod{\phi}{\psi}=\sum_{C\in \B}\phi(C)^*\psi(C)$;
  \item The Euclidian norm is defined as
    $\norm{\phi}=\inprod{\phi}{\phi}$.
  \end{enumerate}
\end{definition}

The Hilbert space $\ell^2=\ell^2(\NN)$ is the standard Hilbert space
of denumerable dimension---all the Hilbert spaces with denumerable
dimension are isomorphic to it. $\ell^2_1$ is the set of the vectors
of $\ell^2$ with unitary norm.

\begin{definition}[computational basis]
  The set of functions 
  $$\cb{\B}=\{\ket{C} : C\in
  \B,\ \ket{C}: \B\to \CC \}$$
  such that for each $C$
   $$
   \ket{C}(D)= \left\{
     \begin{array}{ll}
     1\;\;&\mbox{if}\;C=D\\
     0\;\;&\mbox{if}\;C\neq D 
     \end{array}
   \right. 
   $$
  is called  \emph{computational basis} of $\ell^2(\B)$.
\end{definition}

We can prove that~\cite{RomanBook}:
\begin{theorem}
  The set $\cb{\B}$ is an Hilbert basis of $ \ell^2(\B)$.
\end{theorem}

Let us note that the inner product space $\SPAN{\cb{\B}}$ defined by:
$$
\SPAN{\cb{\B}} =\left\{\sum_{i=1}^n c_i S_i \ |\ c_i\in\CC, S_i\in \cb{\B}, n\in\NN \right\}.
$$
is a proper inner product subspace of $\ell^2(\B)$, but it is not an
Hilbert Space (this means that $\cb{\B}$ is not an Hamel basis of
$\ell^2(\B)$).

The completion of $\SPAN{\cb{\B}}$ is a space isomorphic to  $ \ell^2(\B)$.

By means of a standard result in functional analysis we have:
\begin{theorem}
  \mbox{}
  \begin{enumerate}
  \item $\SPAN{\cb{\B}}$ is a dense subspace of
    $\ell^2(\B)$;
  \item 
    $\ell^2(\B)$ is the (unique! up to isomorphism)
    \emph{completion} of $\SPAN{\cb{\B}}$.
  \end{enumerate}
\end{theorem}

%This fact is important because in the main literature on quantum
%Turing machines, unitary operators are defined by linearity on
%$\SPAN{\cb{\B}}$ not directly on
%$\ell^2(\B)$.  This is not a real problem, since
%it is possible to show that:
\begin{definition}
  Let $\V$ be a complex inner product space, a linear
  application $U:\V\to\V$ is called an
  \emph{isometry} if $\inprod{Ux}{Uy}=\inprod{x}{y}$, for each
  $x,y\in \V$; moreover if $U$ is also surjective, then it is
  called \emph{unitary}.
\end{definition}

Since an isometry is injective, a unitary operator is invertible, and
moreover, its inverse is also unitary.

\begin{definition}
  Let $\V$ be a complex inner product vectorial space, a
  linear application $L:\V\to\V$ is called
  \emph{bounded} if $\exists c>0\;\forall x\; |Lx|\leq c ||x||$.
\end{definition}

\begin{theorem}
  Let $\V$ be a complex inner product vectorial space, for
  each bounded application $U:\V\to\V$ there is one
  and only one bounded application $U^*:\V\to\V$
  s.t. $\inprod{x}{Uy}=\inprod{U^*x}{y}$. We say that $U^*$ is the
  \emph{adjoint} of $U$.
\end{theorem}

It is easy to show that if $U$ is a bounded application, then $U$ is
unitary iff $U$ is invertible and $U^*=U^{-1}$.

\begin{theorem}\label{standard-ext-U}
  Each unitary operator $U$ in
  $\SPAN{\cb{\B}}$ has an unique extension in
  $\ell^2(\B)$~\cite{BerVa97}.
\end{theorem}

\subsection{Dirac notation}
\label{ssec:dirac-notation}

We conclude this brief digest on Hilbert spaces, by a synopsis of the
so-called Dirac notation, extensively used in the paper.

\bigskip

\begin{center}
  \begin{tabular}{|c|c|}
    \hline
    \textsf{mathematical notion} & \textsf{Dirac notation}\\
    \hline
    inner product $\inprod{\phi}{\psi}$ & $\kinprod{\phi}{\psi}$\\
    \hline
    vector $\phi$ & $\ket{\phi}$
    \\
    \hline
    dual of vector $\phi$ & $\tek{\phi}$\\
    i.e., the linear application $d_\phi$ &
    \\
    defined as $d_\phi(\psi)=\inprod{\phi}{\psi}$ & 
                                                    note that $\kinprod{\phi}{\psi}=\tek{\phi}(\ket{\psi})$
    \\
    \hline
  \end{tabular}
\end{center}

\bigskip

%Let $L',L''$ a linear application, in order to avoid the proliferation of parentheses we will write $<L'\phi |\psi>, <\phi|L''\psi>, <L'\phi|L''\psi>$ instead of $<L'\ket{\phi} |\psi>, <\phi|L''\psi>, <L'\phi|L''\psi>$
Let $L$ be a linear application, with $\tek{\phi}L\ket{\psi}$ we denote $\kinprod{\phi}{L\psi}$.

%% file: appendix.tex
% -*- Mode: LaTeX; TeX-master: "bozza-4.0.tex" -*- 
% Time-stamp:  <2015-04-10 13:11:53 guerrini>

\section{Proof of Theorem~\ref{thm:time-evol-unitary}}
\label{sec:proof-local-cond}

In this section we shall give the proof of
Theorem~\ref{thm:time-evol-unitary} in full details. Following
Bernstein and Vazirani \cite{BerVa97} and Nishimura and Ozawa
\cite{NishOza10}, we shall prove a stronger result indeed, that the
local unitary conditions are not only sufficient to obtain a unitary
time evolution, but also necessary
(Theorem~\ref{thm:qtm-iff-unitary}).

\subsection{Pre-QTM}
\label{sec:pre-qtm}

A pre-QTM is a tuple
$M=\langle \Sigma, \Q, \Q_s, \Q_t, \delta, q_i, q_f\rangle$ for which all
the requirements demanded for a QTM (Definition~\ref{def:QTM}) but the
local unitary conditions hold. In other words, while $\delta_s$ and
$\delta_t$ are defined and constrained as for QTM's, we do not have
any condition on $\delta_0$. Ground and quantum configurations, step
function and time evolution operator of pre-QTM's are defined as for
QTM's.
In the following, $M$ will denote a pre-QTM.

\subsection{Basic results}

Let us start with some basic results which hold in any Hilbert space.
As usual, if $U$ is an operator, $U^*$ denotes its adjoint.

\begin{lemma}\label{lem:basic-results}
  Let $U:\ell^2(\GConf_M)\to \ell^2(\GConf_M)$. 
  \begin{enumerate}
  \item $U$ is an isometry iff it is left-invertible and its adjoint
    is its left-inverse, that is, $U^*U=1$.
  \item If $U$ is an isometry, then $P=UU^*$ is an orthonormal
    projection and
    \begin{enumerate}
    \item $\norm{P\ket{\phi}} \leq \norm{\ket{\phi}}$, for every $\phi$;
    \item $\norm{P\ket{\phi}} = \norm{\ket{\phi}}$ iff $P\ket{\phi} =
      \ket{\phi}$;
    \item $P=1$ iff
      $\norm{P\ket{C}}=\tek{C} P \ket{C}=1$, for every
        $C\in\GConf_M$.
    \end{enumerate}
  \end{enumerate}
\end{lemma}
\begin{proof}\mbox{}
  \begin{enumerate}
  \item $\langle U \ket{\phi}, U \ket{\psi}\rangle = \langle U^*U
    \ket{\phi}, \ket{\psi}\rangle = \kinprod{\phi}{\psi}$ iff
    $U^*U=1^*=1$.
  \item $P$ is a projection when $PP=P$ and, moreover, it is
    orthonormal when it is self-adjoint, that is $P^*=P$. Equivalently, $P$ is
    a projection iff, for every $\phi$, we have
    $\ket{\phi}=\ket{\phi_1}+\ket{\phi_0}$ for some (unique) $\phi_1$
    and $\phi_0$ s.t.\ $P\ket{\phi_1}=\phi_1$ and $P\ket{\phi_0}=0$;
    it is orthonormal when
    $\kinprod{\phi_1}{\phi_0}=0$ indeed. 

    $P = UU^*$ is clearly self-adjoint. Moreover,
    $PP=UU^*UU^*=UU^*=P$, since $U$ is an isometry and then $U^*U=1$.
    Thus, $P$ is an orthonormal projection.

    Let $\ket{\phi}=\ket{\phi_1}+\ket{\phi_0}$ be an orthonormal
    decomposition as above.
    $\norm{\ket{\phi}} = \norm{\ket{\phi_1}}+2\Re
    \kinprod{\phi_1}{\phi_0}+ \norm{\ket{\phi_0}} =
    \norm{P\ket{\phi_1}}+\norm{\ket{\phi_0}} =
    \norm{P\ket{\phi}}+\norm{\ket{\phi_0}}$.
    \begin{enumerate}
    \item $\norm{P\ket{\phi}} \leq  \norm{P\ket{\phi}}+\norm{\ket{\phi_0}}
      = \norm{\ket{\phi}}$
    \item $\norm{P\ket{\phi}}=\norm{\ket{\phi}}$ iff
      $\norm{\ket{\phi_0}}=0$, that is, iff $\ket{\phi_0}=0$.
    \item By the previous item, $P\ket{\phi}=\ket{\phi}$ for every
      $\phi$, iff $\norm{P\ket{\phi}}=\norm{\ket{\phi}}$ for every
      $\phi$, that is, iff $\norm{P\ket{C}}=\norm{\ket{C}}=1$ for
      every $C\in\GConf_M$, that is, iff $\norm{P\ket{C}}=\langle
      P\ket{C},P\ket{C}\rangle=\tek{C}P\ket{C}=1$, for
      every $C\in\GConf_M$.
    \end{enumerate}
  \end{enumerate}
\end{proof}

\subsection{Reverse transitions}

The reverse step function of $M$ is defined by
$$\rgamma_{\Sigma,\Q}(\langle \alpha v_R, p, wv_L \beta \rangle, q,u,d) \simeq
\begin{cases}
  \langle \alpha v_R, q, u\beta \rangle & \qquad\mbox{when $d = L$} \\
  \langle \alpha, q, uwv_L\beta \rangle & \qquad\mbox{when $d = R$}
\end{cases}
$$

Let us say that an $R$ or $L$ step of $\rgamma$ is an $R$-reverse or
$L$-reverse step. We see that an $R/L$-reverse step of $\rgamma$
revert an $R/L$ step of the step function $\gamma$. While in
$\gamma$ both the $L$-step and the $R$-step replace the same symbol,
the current symbol of the configuration, in $\rgamma$ the symbols
replaced by the $R$-reverse step and by the $L$-reverse step are in
different positions. We have then a current $L$-reverse symbol $v_L$
and a current $R$-reverse symbol $v_R$.

\begin{lemma}\label{lem:rev-step}
  Let $C\simeq\langle\alpha_1w_1,q,u\beta_1\rangle$ and
  $D\simeq\langle\alpha_2v_R,p,w_2v_L\beta_2\rangle$. 
  $$\rgamma(D,q,u,d) = C
  \qquad\mbox{iff}\qquad \gamma(C,p,v_d,d) = D$$
\end{lemma}
\begin{proof}
  When $d=R$, we have $C \simeq
    \langle\alpha_2,q,uw_2v_L\beta_2\rangle \simeq \rgamma(D,q,u,R)$ iff
    $\alpha_1w_1\simeq_l\alpha_2$ and
    $w_2v_L\beta_2\simeq_r\beta_1$ iff $\gamma(C,p,v_R,R) \simeq
    \langle\alpha_1w_1v_R,q,\beta_1\rangle \simeq D$.
    When $d=L$, we have $C \simeq
    \langle\alpha_2v_Rw_2,q,u\beta_2\rangle \simeq
    \rgamma(D,q,u,L)$ iff $\alpha_1\simeq_l\alpha_2v_R$ and
    $w_1=w_2$ and $\beta_2\simeq_r\beta_1$ iff $\gamma(C,p,v_L,L)
    \simeq \langle\alpha_1,q,w_1v_L\beta_1\rangle \simeq D$.
\end{proof}

\begin{lemma}\label{lem:rgamma-gamma}
  Let $C[q,u]$ be the configuration obtained by substituting the state
  $q$ and the symbol $u$ for the current state and the current symbol
  of the configuration $C$, and let
  $C_{p,v,d_1}^{q,u,d_2} = \rgamma(\gamma(C,p,v,d_1), q,u,d_2)$.
  \begin{enumerate}
  \item If $d=d_1=d_2$, then $C[q,u]=C^{p,v,d}_{q,u,d}$, for every
    $(p,v)\in \Q\times\aball{\Sigma}$.
  \item When $d_1\neq d_2$, there is $C_{d_1d_2}[q,u,v]$, which does not
    depend on $p$, s.t.\ $C_{d_1d_2}[q,u,v]=C_{p,v,d_1}^{q,u,d_2}$,  for every
    $p\in \Q$.
  \end{enumerate}
  Moreover, let
  $(q',u',v'),(q'',u'',v'')\in \Q
  \times\aball{\Sigma}\times\aball{\Sigma}$
  and $d_1, d_2,d_1',d_2'\in\DD_x$ with $d_1\neq d_2$ and
  $d_1'\neq d_2'$. If the tape of the configuration $C$ contains at
  least a non-empty cell in addition to the current one (that is,
  $C\neq\langle\lambda,q,u\lambda\rangle$), then
  \begin{enumerate}\setcounter{enumi}{2}
  \item $C[q',u'] = C[q'',u'']$ iff $(q',u')=(q'',u'')$;
  \item $C[q',u'] \neq C_{d_1d_2}[q'',u'',v'']$;
  \item $C_{d_1d_2}[q',u',v'] = C_{d_1'd_2'}[q'',u'',v'']$ iff
    $(d_1,d_2)=(d_1',d_2')$ and $(q',u',v')=(q'',u'',v'')$
  \end{enumerate}
\end{lemma}
\begin{proof}
  Let $C \simeq \langle \alpha z_lw_l,q,uw_rz_r \beta \rangle$. By
  computing $C_{p,v,d_1}^{q,u,d_2}$, we see that (1) and (2) hold with
  \begin{align*}
    C[q',u'] &\simeq \langle \alpha z_lw_l, q', u'w_rz_r \beta \rangle\\
    % C_l[p,v] &=  \langle \alpha z_l, p, w_lvw_rz_r \beta \rangle\\
    % C_r[p,v] &=  \langle \alpha z_lw_lv, p, w_rz_r \beta \rangle\\
    C_{LR}[q',u',v] &\simeq \langle \alpha, q', u'w_lvw_rz_r \beta \rangle\\
    C_{RL}[q',u',v] &\simeq \langle \alpha z_lw_lvw_r, q', u' \beta \rangle
  \end{align*}
  From which, we can prove the following items.
  \begin{enumerate}\setcounter{enumi}{2}
  \item Immediate.
  \item Let $C[q',u'] \simeq C_{LR}[q'',u'',v'']$. Then,
    $\alpha z_lw_l \simeq_l \alpha$ and $q'=q''$ and
    $u'w_rz_r\beta \simeq_R u''w_lv''w_rz_r\beta$. Which is possible
    only if $\alpha\simeq\beta\simeq\lambda$ and
    $z_l=z_r=w_l=w_r=v''=\Box$ and $u'=u''$. But this is the case only
    if $C=\langle\lambda,q,u\lambda\rangle$. And analogously for
    $C[q',u'] \simeq C_{RL}[q'',u'',v'']$, it may hold only if
    $C=\langle\lambda,q,u\lambda\rangle$.
  \item The case $(d_1,d_2)=(d_1',d_2')$ is immediate. Thus, let us
    assume $C_{LR}[q',u',v'] \simeq C_{RL}[q'',u'',v'']$. We see that
    this is possible only if $\alpha\simeq_l\alpha z_lw_lv''w_r$ and
    $q'=q''$ and $u'w_lv'w_rz_r\beta\simeq_ru''\beta$. Which holds
    only if $\alpha\simeq\beta\simeq\lambda$ and
    $z_l=z_r=w_l=w_r=v'=v''=\Box$ and $u'=u''$. That implies
    $C=\langle \lambda,q,u\lambda \rangle$.
  \end{enumerate}
\end{proof}

\begin{lemma}\label{lem:gamma-rgamma}
  Let
  $\overline{C}_{p,v,d_2}^{q,u,d_1} =
  \gamma(\rgamma(C,q,u,d_1),p,v,d_2)$.
  \begin{enumerate}%\setcounter{enumi}{5}
  % \item For $d=\DD_x$, there is
  %   $\overline{C}_d[p',v'] = \overline{C}_{p',v',d}^{q,u,d}$ s.t.\
  %   $C=\overline{C}_d[p,v_d]$, for some
  %   $(p,v_d)\in Q\times\aball{\Sigma}$. Moreover,
  %   $\overline{C}_d[p,v_d] = \overline{C}_d[p',v']$ iff $(p,v_d)=(p',v')$.
  \item There is a triple
    $(q,w_L,w_R)\in \Q\times\aball{\Sigma}\times\aball{\Sigma}$, s.t.,
    for any $(q',u')\in Q\times\aball{\Sigma}$ and $d=\DD_x$, we have
    $C = \overline{C}_{p,v,d}^{q',u',d}$ iff $(p,v)=(q,w_d)$.
  \item For $d_1\neq d_2 \in \DD_x$,
    $C=\overline{C}_{p,v,d_1}^{q',u',d_2}$, iff
    $C=\langle \lambda, q, \lambda\rangle$ and $u'=v=\lambda$ and
    $p=q$.
  \end{enumerate}
\end{lemma}
\begin{proof}
  \mbox{}
  \begin{enumerate}%\setcounter{enumi}{5}
  \item 
    Let us define
    $\overline{C}_R[q',u']=\langle \alpha z_1u', q',
    uw_Lz_2\beta\rangle$
    (the configuration obtained by replacing $q'$ and $u'$ for the
    current state $q$ and the reverse right symbol $w_R$ of $C$) and
    $\overline{C}_L[q',u']=\langle \alpha z_1w_R, q',
    uu'z_2\beta\rangle$
    (the configuration obtained by replacing $q'$ and $u'$ for the
    current state $q$ and the reverse right symbol $w_R$ of $C$), it
    is readily seen that $C_{q',u',d}^{p,v,d}=C_d[p,v]$ and, as a
    consequence, $C=C_d[p,v]$ iff $(p,v)=(q,u)$.
  \item By direct computation, we see that
    $\overline{C}_{p,v,L}^{q',u',R} = \langle \alpha z_lw_Ruu'v, p,
    \beta\rangle$
    and
    $\overline{C}_{p,v,R}^{q',u',L} = \langle \alpha, p,
    vu'uw_Lz_R\beta\rangle$.
    From which, we see that $C=\overline{C}_{p,v,L}^{q',u',R}$ iff
    $\alpha=\beta=\lambda$ and $z_L=z_R=w_L=w_R=u=u'=v=\Box$ and
    $C=\overline{C}_{p,v,R}^{q',u',L}$ iff
    $z_L=z_R=w_L=w_R=u=u'=v=\Box$.
  \end{enumerate}
\end{proof}

\subsection{The adjoint of $U_M$} 

We can now compute the adjoint of the operator $U_M$. For this, we
have already given the reverse transition $\rgamma$, but we also need
to reverse the quantum transition function $\delta$. For
$x\in\{0,s,t\}$, let us take
$$\rdelta_x:\T_x\to \ell^2(\varS_x\times\DD_x)
\qquad\mbox{s.t.}\qquad \rdelta_x(p,v)(q,u,d) =
\delta_x(q,u)(p,v,d)^*$$
where $\DD_0=\DD$ and $\DD_s=\DD_t=\{R\}$. Then, let us
define
$$\TConf{x}_M = \{\langle\alpha v, p, \beta\rangle \in \GConf_M \mid
(p,v)\in \T_x\}$$

It is readily seen that $\TConf{0}_M$, $\TConf{s}_M$ and $\TConf{t}_M$
are a partition of $\GConf$, since they are pairwise disjoint and
$\GConf = \TConf{0}_M \cup\TConf{s}_M \cup \TConf{t}_M$. Moreover, given
$$C\simeq\langle \alpha v_R, p, wv_L\beta\rangle\in\TConf{x}_M$$
we have that, $(p,v_R)\in \T_x$ by definition, and when $x=0$, that
$(p,v_L)\in \T_0$ also (since by the definition of ground
configuration, $p\in \Q_0\cup \Q_t$ implies that
$wv_L\beta\in\Sigma^*$). Thus, we can finally define
\begin{align*}
  W^*_M\ket{C} & = \quad 
                 \sum_{(q,u,d)\in \varS_x\times\DD_x}
                 \rdelta_x(p,v_d)(q,u,d)\,\ket{C_{q,u,d}}  \\
               & = \quad 
                 \sum_{(q,u,d)\in \varS_x\times\DD_x}
                 \delta_x(q,u)(p,v_d,d)^*\,\ket{C_{q,u,d}}
\end{align*}

\noindent
where $C_{q,u,d}=\rgamma_{\Sigma,\Q}(C,q,u,d)$.

We remark that, w.r.t.\ the definition of $W_M$, in the range of the
sum in $W_M^*$, we have $\varS_x\times\DD_x$ in the place of
$\T_x\times\DD$. In the case $x=0$, there is no difference, since
$\DD_x=\DD$ and then we consider both the reverse deplacements; in the
case $x=s,t$ instead, $\DD_x=\{R\}$ and then we consider the $R$
deplacement only. Technically, this is necessary as in these cases
$(p,v_L)\not\in \T_x$, and then $\delta(q,u)(p,v_L,d)$ would not be
defined. Indeed , for $x=s,t$, any configuration in $\TConf{x}_M$ can
be entereded from an $R$ deplacement only; then, in this case, to
reverse the quantum transition functions it suffices to consider $R$
reverse deplacements only. On the other hand, even in the definition
of $W_M$, we might have restricted the sum to the $R$ deplacement
only, in the case $x=s,t$. Indeed,
$$W_M(\ket{C}) = \sum_{(p,v,d)\in \T_x\times\DD_x}
\delta_x(q,u)(p,v,d)\,\ket{C_{p,v,d}}
$$
for $C=\langle\alpha,q,u\beta\rangle\in\SConf{x}_M$, and $x=0,s,t$.

\begin{lemma}\label{lem:WMstar-welldef}
  $W_M^*\ket{C} \in \SPAN{\cb{\GConf}}$, for every
  $C\in\GConf_M$. Then, $W_M^*$ defines an automorphism 
  $$W_M^*:\SPAN{\cb{\GConf}}\to\SPAN{\cb{\GConf}}$$
  of the linear space $\SPAN{\cb{\GConf}}$.
\end{lemma}
\begin{proof}
  By case analysis, as in the proof of
  Proposition~\ref{prop:WM-welldef}.
\end{proof}

We can also see that $W_M^*$ maps every $\TConf{x}_M$ into
$\SConf{x}_M$, which is indeed the converse of the fact that $W_M$
maps $\SConf{x}_M$ into $\TConf{x}_M$.

\begin{lemma}\label{lem:W-images}
  \mbox{}
  \begin{enumerate}
  \item $W_M({\SPAN{\cb{\SConf{x}_M}}}) \subseteq \SPAN{\cb{\TConf{x}_M}}$
  \item $W_M^*({\SPAN{\cb{\TConf{x}_M}}}) \subseteq
    \SPAN{\cb{\SConf{x}_M}}$.
  \end{enumerate}
\end{lemma}
\begin{proof}
  By case analysis.
\end{proof}

We can now prove that $W_M^*$ defines the adjoint of $U_M$.

\begin{lemma}
  The unique extension of $W_M^*$ to the Hilbert space
  $\ell^2(\GConf_M)$ is the adjoint $U_M^*$ of $U_M$.
\end{lemma}
\begin{proof}
  It suffices to prove that
  $\inprod{W_M^* \ket{D}}{\ket{C}} = \tek{D} W_M\ket{C}$, for every
  $C,D\in\GConf_M$. Let $C\in\SConf{x}$ and
  $D\in\TConf{y}$, with $x,y\in\{0,s,t\}$. Since
  $\inprod{W_M^*\ket{D}}{\ket{C}} = \tek{D}W_M\ket{C} = 0$ if
  $x\neq y$ (by Lemma~\ref{lem:W-images}), we have to
  analyse the cases $x=y$ only.

  Let $C\simeq\langle\alpha_1w_1,q,u\beta_1\rangle\in\SConf{x}$,
  $D\simeq\langle\alpha_2v_R,p,w_2v_L\beta_2\rangle\in\TConf{x}$,
  $D_{q,u,d}=\rgamma(D,q,u,d)$, and $C_{p,v,d}=\gamma(C,p,v,d)$.
  Let us start with the case $x=0$.
  
  \begin{align*}
   \inprod{ W^*_M \ket{D}}{\ket{C}} & = \quad 
                                      \sum_{(q',u',d)\in \varS_x\times\DD_x}
                                      \rdelta_x(p,v_d)(q',u',d)^*\,\kinprod{D_{q',u',d}}{C}   \\
                                    & = \quad
                                      \sum_{(q',u',d)\in \varS_x\times\DD_x}
                                      \delta_x(q',u')(p,v_d,d)\,\kinprod{D_{q',u',d}}{C} \\
    & = \quad 
    \sum_{d\in\DD_x} \delta_x(q,u)(p,v_d,d)\,\kinprod{D_{q,u,d}}{C} \\
    & = \quad 
    \sum_{d\in\DD_x} \delta_x(q,u)(p,v_d,d)\,\kinprod{D}{C_{p,v_d,d}} \\
    & = \quad
    \sum_{(p',v',d)\in \T_x\times\DD_x}
    \delta_x(q,u)(p',v',d)\,\kinprod{D}{C_{p',v',d}} \\
    & = \quad \tek{D} W_M\ket{C} 
  \end{align*}

  Where we have used the facts that $\kinprod{D_{q',u',d}}{C}=0$ if
  $(q',u')\neq(q,u)$, that $\kinprod{D}{C_{p',v',d}}=0$ if
  $(p',v')\neq(p,v_d)$ (by inspection, we see that in these cases the
  above configurations differ for the current state or for the current
  symbol), and that $\kinprod{D_{q,u,d}}{C}=\kinprod{D}{C_{p,v_d,d}}$,
  by Lemma~\ref{lem:rev-step}.
\end{proof}

\subsection{Unitarity of the time evolution operator}

In the following, we shall complete the proof that the time evolution
operator of a QTM is unitary. Firstly, we show that the time evolution
operator $U_M$ of a pre-QTM $M$ is an isometry iff the local unitary
conditions hold (Lemma~\ref{lem:isometry-cond}), then we shall see
that, in the particular case of pre-QTMs, $U_M$ is unitary when it is
is an isometry (Lemma~\ref{lem:isometry-imp-unitary}). Therefore, a
pre-QTM is a QTM iff its time evolution operator is unitary.

\begin{lemma}\label{lem:isometry-cond}
  $U_M^*U_M=1$ iff the local unitary conditions holds, that is, $U_M$
  is an isometry iff $M$ is a QTM.
\end{lemma}
\begin{proof}
  It suffices to prove that, for $C\in\SConf{x}_M$, with $x\in\{0,s,t\}$,
  $$W_M^*W_M\ket{C} \quad  = \sum_{(p,v,d)\in \T_x\times\DD}
  \delta_x(q,u)(p,v,d)\,W_M^*\ket{C_{p,v,d}} \quad = \quad 1$$
  iff the local unitary conditions hold. 

  Let us assume
  $$C \simeq \langle \alpha v_Rw,q,uzv_L \beta \rangle\in\SConf{x}$$ 
  $$C_{p,v,d} = \gamma(C,p,v,d)
  \qquad\qquad
  C_{p,v,d}^{q',u',d'} = \rgamma(\gamma(C,p,v,d), q',u',d')$$
  we have
  $$C_{p,v,L}  \simeq \langle \alpha v_R,q,wvzv_L \beta \rangle
  \qquad\qquad
  C_{p,v,R}  \simeq \langle \alpha v_Rwv,q,zv_L \beta \rangle$$

  The cases $x=s$ and $x=t$ are immediate. For instance, if $x=t$, we
  have $(q,u)\in \Q_t\times\Sigma$, and $\delta_t(q,u)(p,v,d)=1$, when
  $(p,v,d)=(q,\symextra{u}, R)$, while $\delta_t(q,u)(p,v,d)=0$,
  otherwise. Thus
  $$W_M^*W_M\ket{C} \quad  = \quad W_M^*\ket{C_{q,\symextra{u},R}} \quad =
  \sum_{(q',u')\in \varS_t}\delta_t(q',u')(q,\symextra{u},R)^* \ket{C^{q',u',R}_{q,\symextra{u},R}}$$
  but $\delta_t(q',u')(p,\symextra{u},R)=1$,
  $(p,u)=(q',u')$, and
  $\delta_t(q',u')(p,\symextra{u},r)=0$, otherwise; thus
  $$W_M^*W_M\ket{C} \quad  \quad =\quad  \ket{C^{q,u,R}_{q,\symextra{u},R}} \quad 
  = \quad \ket{C}$$
  (by Lemma~\ref{lem:rgamma-gamma}) and analogously for the case
  $x=s$. Therefore, for $x\in\{s,t\}$, since the equivalence holds for
  every pre-QTM. 

  Let us now analyse the case $x=0$.
  By definition,
  $$W_M^*\ket{C_{p,v,d}}=\sum_{(q',u',d')}=
  \delta_0(q',u')(p,v'_{d'},d')^*\,\ket{C^{q',u',d'}_{p,v,d}}$$
  for some $v'_{d'}$ that depend on $C_{p,v,d}$. By reorganizing the
  sums according to the cases $d'=d$ and $d'\neq d$, we get
  \begin{multline*}
    W_M^*W_M\ket{C} \quad  = \sum_{(p,v,d)\in \T_x\times\DD}
    \delta_0(q,u)(p,v,d)\,W_M^*\ket{C_{p,v,d}} \\
    \begin{aligned}
      \qquad & = \sum_{(p,v,d)\in \T_x\times\DD}\sum_{(q',u')\in \varS_x}
      \delta_0(q,u)(p,v,d)\, 
      \delta_0(q',u')(p,v,d)^*\,\ket{C^{q',u',d}_{p,v,d}}   \\
      & \quad + \sum_{(p,v)\in \T_x}\sum_{(q',u')\in \varS_x}
      \delta_0(q,u)(p,v,L)\, 
      \delta_0(q',u')(p,v_R,R)^*\,\ket{C^{q',u',R}_{p,v,L}}   \\
      & \quad + \sum_{(p,v)\in \T_x}\sum_{(q',u')\in \varS_x}
      \delta_0(q,u)(p,v,R)\, 
      \delta_0(q',u')(p,v_L,L)^*\,\ket{C^{q',u',L}_{p,v,R}}   \\
    \end{aligned}
  \end{multline*}
  By Lemma~\ref{lem:rgamma-gamma}, this corresponds to
  \begin{multline*}
    W_M^*W_M\ket{C} \\
    \begin{aligned}
      \quad & =
      \sum_{(p,v,d)\in \T_0\times\DD}
      |\delta_0(q,u)(p,v,d)|^2\,\ket{C}   \\
      &\ +
      \sum_{(q',u')\in \varS_0\setminus\{(q,u)\}}\;
      \sum_{(p,v,d)\in \T_0\times\DD}
      \delta_0(q,u)(p,v,d)\,
      \delta_0(q',u')(p,v,d)^*\,\ket{C[q',u']}   \\
      & \ +
      \sum_{(q',u')\in \varS_0}\;
      \sum_{v\in \Sigma}\;
      \sum_{p\in \Q_0\cup \Q_t}
      \delta_0(q,u)(p,v,R)\, 
      \delta_0(q',u')(p,v_L,L)^*\,\ket{C_{RL}[q',u',v]}   \\
      & \ +
      \sum_{(q',u')\in \varS_0}\;
      \sum_{v\in \Sigma}\;
      \sum_{p\in \Q_0\cup \Q_t}
      \delta_0(q,u)(p,v,L) 
      \delta_0(q',u')(p,v_R,R)^*\,\ket{C_{LR}[q',u',v]} 
    \end{aligned}
  \end{multline*}
  where, $C$, $C[q'u']$, $C_{LR'}[q',u',v]$, and $C_{RL}[q',u',v]$ are
  never equal, does not depend on $(p,v)$, and $C$, $C[q'u']$ are
  independent from $d$ also.
  
  From which we may conclude that, $W_M^*W_M\ket{C} = \ket{C}$ for
  every $C\in\SConf{0}$, iff for every $q,q'\in \Q_0\cup \Q_s$ and
  $v,u,u'v_L,v_R\in\Sigma$

  \begin{align*}
    1 &= 
        \sum_{(p,v,d)\in \T_0\times\DD}
        |\delta_0(q,u)(p,v,d)|^2
    \\
    0 &=
        \sum_{(p,v,d)\in \T_0\times\DD}
        \delta_0(q,u)(p,v,d)\,
        \delta_0(q',u')(p,v,d)^*
    \\
    0  &=
         \sum_{p\in \Q_0\cup \Q_t}
         \delta_0(q,u)(p,v,R)\, 
         \delta_0(q',u')(p,v_L,L)^*
    \\
    0  &=
         \sum_{p\in \Q_0\cup \Q_t}
         \delta_0(q,u)(p,v,L) 
         \delta_0(q',u')(p,v_R,R)^*
  \end{align*}
  that is, iff the local unitary conditions hold.
\end{proof}

\begin{lemma}\label{lem:isometry-imp-unitary}
  If $U_M$ is an isometry, then $U_MU_M^*=1$. As a consequence, $U_M$
  is unitary.
\end{lemma}
\begin{proof}
  Since $U_M$ is an isometry, by Lemma~\ref{lem:basic-results} we
  know that $U_MU_M^*$ is an orthonormal projection and $U_MU_M^*=1$
  iff $\tek{C}U_MU_M^*\ket{C}=\tek{C}W_MW_M^*\ket{C}=1$ for every
  $C\in\GConf_M$.

  Let us assume that
  $$C \simeq \langle \alpha v_R,p,wv_L\beta \rangle\in\TConf{x}$$ 
  $$\overline{C}_{q,u,d} = \gamma(C,q,u,d)
  \qquad\qquad
  \overline{C}_{p',v',d'}^{q,u,d} = \gamma(\rgamma(C,q,u,d), p',v',d')$$
  we have
  $$\overline{C}_{q,u,L}  \simeq \langle \alpha v_Rw, q, u \beta \rangle
  \qquad\qquad
  \overline{C}_{q,u,R}  \simeq \langle \alpha, q, u wv_L\beta \rangle$$

  \begin{multline*}
    \tek{C} W_MW_M^*\ket{C} \\
    \begin{aligned}
      \quad & = 
      \sum_{(q,u,d)\in \varS_x\times\DD_x} \delta_x(q,u)(p,v_d,d)^*\,\tek{C}W_M\ket{\overline{C}_{q,u,d}} 
      \\
      \quad & = 
      \sum_{\dind{(q,u,d)\in \varS_x\times\DD_x}{(p',v',d')\in \T_x\times\DD_x}}
      \delta_x(q,u)(p,v_d,d)^*\delta_x(q,u)(p',v',d')\,
      \kinprod{C}{\overline{C}_{p',v',d'}^{q,u,d}}
    \end{aligned}
  \end{multline*}

  By Lemma~\ref{lem:gamma-rgamma}, we have that
  \begin{enumerate}
  \item\label{item:rgginpro:eq} when $d=d'$
    $$\kinprod{C}{\overline{C}_{p',v',d'}^{q,u,d}} =
    \kinprod{C}{\overline{C}_{p',v',d}^{q,u,d}} =
    \begin{cases}
      1 & \qquad\mbox{if $(p',v')=(p,v_d)$} \\
      0 & \qquad\mbox{otherwise}
    \end{cases}$$
  \item\label{item:rgginpro:diff} when $d\neq d'$
    $$\kinprod{C}{\overline{C}_{p',v',d'}^{q,u,d}} =
    \begin{cases}
      1 & \qquad
      \begin{minipage}[t]{5cm}
        if $C=\langle\lambda, p, \lambda\rangle$ and $p=p'$\\
        and $v_L=v_R=v'=u=\Box$
      \end{minipage}\\
      0 & \qquad\mbox{otherwise}
    \end{cases}
    $$
  \end{enumerate}

  When $x\in\{s,t\}$, we have $\DD_x=\{R\}$ and then, by
  item~(\ref{item:rgginpro:eq}),
  $$\tek{C} W_MW_M^*\ket{C} \quad=\quad 
  \sum_{(q,u)\in \varS_x} |\delta_x(q,u)(p,v_R,R)|^2$$
  Morevover, in both cases $\delta_s(q,u)(p,v_R,R) = 0$, with the
  exception of the case $(q,u)=(p,\widetilde{v}_R)$, where
  \begin{enumerate}
  \item for $x=s$, we have that $v_R=a\in\Sigma$ and
   $\widetilde{v}_R=\symextra{a}$;
  \item for $x=t$, we have that $v_R=\symextra{a}\in\abextra{\Sigma}$ and
    $\widetilde{v}_R=a$.
  \end{enumerate}
  Since $\delta_s(p,\widetilde{v}_R)(p,v_R,R) = 1$, in both cases we
  get $\tek{C} W_MW_M^*\ket{C}=1$.

  Let us now consider the case $x=0$. 
  \begin{multline*}
    \tek{C} W_MW_M^*\ket{C} \\
    \begin{aligned}
      \qquad
      & = 
      \sum_{(q,u,d)\in \varS_0\times\DD_0}
      |\delta_0(q,u)(p,v_d,d)|^2
      \\
      & \quad +
      \sum_{(q,d)\in \varS_0\times \DD}
      \delta_0(q,\Box)(p,\Box,d)^*\,\delta_x(q,\Box)(p,\Box,\widetilde{d})
      \kinprod{C}{\overline{C}_{p,\Box,\widetilde{d}}^{q,\Box,d}}
    \end{aligned}
  \end{multline*}
  where $\widetilde{d}=R$ when $d=L$, and $\widetilde{d}=L$ when
  $d=R$.
  We remark that the last addend may not be equal to $0$ only when
  $C=\langle\lambda,p,\lambda\rangle$ (see item~(\ref{item:rgginpro:diff}) above).

  In order to analyse it, let us not consider a single configuration
  only, but a whole family of configurations
  $$C[p',v_R',v_L'] \simeq \langle \alpha v_R', p', wv_L'\beta \rangle \in \TConf{0}$$
  which differ from
  $C\simeq\langle \alpha v_R, p, wv_L\beta \rangle\in\TConf{0}$ for the
  current state $p$ and the current $R$-reverse and $L$-reverse
  symbols $v_R$ and $v_L$, respectively. More precisely, we take
  $$\B_{C} = \{ C[p',v_R',v_L']  \in \TConf{0}  \mid (p',v_R',v_L')\in Q\times\aball{\Sigma}^2\}$$ 
  which, by the definition of $\TConf{0}$ corresponds also to
  $$\B_{C} = \{ C[p',v_R',v_L'] \mid
  (p',v_R',v_L')\in (\Q_0\cup \Q_t)\times\Sigma^2\}$$
  From which, it is readily seen that
  $$|\B_C|=(|\Q_0 \cup \Q_t|)\,|\Sigma|^2$$  

  Finally, let us take
  \begin{multline*}
    \sum_{C'\in \B_c}\tek{C'} W_MW_M^*\ket{C'}  \\
    \begin{aligned}
      \qquad & = 
      \sum_{(p',v_R',v_L')\in (\Q_0 \cup \Q_t)\times\Sigma^2}
      \tek{C[p',v_R',v_L']} W_MW_M^*\ket{C[p',v_R',v_L']} \\
      & = 
      \sum_{v_L'\in\Sigma}
      \sum_{(q,u)\in \varS_0}
      \sum_{(p',v_R')\in \T_0}
      |\delta_0(q,u)(p',v_R',R)|^2
      \\
      & \quad + 
      \sum_{v_R'\in\Sigma}
      \sum_{(q,u)\in \varS_0} 
      \sum_{(p',v_L')\in \T_0}
      |\delta_0(q,u)(p',v_L',L)|^2
      \\
      & \quad + 
      \sum_{\dind {(v_R',v_L')\in\Sigma^2}{(q,d)\in (\Q_0\cup \Q_s)\times\DD}}\;
      \sum_{p'\in \Q_0\cup \Q_t}
      \delta_0(q,\Box)(p',\Box,d)^*\,\delta_0(q,\Box)(p',\Box,\widetilde{d})\kinprod{C}{\overline{C}_{q,\Box,d}^{p',\Box,\tilde{d}}}
    \end{aligned}
  \end{multline*}

  By the local unitary conditions,
  $$\sum_{p'\in \Q_0\cup \Q_t}
  \delta_0(q,\Box)(p',\Box,d)^*\,\delta_0(q,\Box)(p',\Box,\widetilde{d}) = 0$$
  for every $q\in \Q_0\cup \Q_s$. Therefore, by taking into account that 
  $\delta_0(q,u)(p',v_R',R)$ does not depend on $v_L'$, and that $\delta_0(q,u)(p',v_L',L)$ does not depend on $v_R'$, we have  

  \begin{align*}
    \sum_{C'\in \B_c}\tek{C'} W_MW_M^*\ket{C'}  
    & = 
      |\Sigma| \sum_{(q,u)\in \varS_0} \sum_{(p',v_R')\in \T_0}
      |\delta_0(q,u)(p',v_R',R)|^2
    \\
    & \qquad +
      |\Sigma| \sum_{(q,u)\in \varS_0}
      \sum_{(p',v_L')\in \T_0}
      |\delta_0(q,u)(p',v_L',L)|^2
    \\
    & = 
      |\Sigma| \sum_{(q,u)\in \varS_0}
      \sum_{(p',v',d)\in \T_0\times\DD}
      |\delta_0(q,u)(p',v',d)|^2
      \\
  \end{align*}

  But by the local unitary conditions,
  $\sum_{(p',v',d)\in \T_0\times\DD} |\delta_0(q,u)(p',v',d)|^2 = 1$ for
  every $(q,u)\in \varS_0$. Thus, 

  \begin{equation*}
    \sum_{C'\in \B_c}\tek{C'} W_MW_M^*\ket{C'} 
    \ =\ 
    |\Sigma| \sum_{(p'',v'')\in \varS_0} 1
    \ =\ 
    |\Q_0\cup \Q_t|\,|\Sigma|^2
    \ =\ 
    |\B_C|
  \end{equation*}
  
  Finally, let us recall that, by Lemma~\ref{lem:basic-results},
  $\tek{C} W_MW_M^*\ket{C} \leq 1$ for every
  $C\in\GConf_M$. Therefore, for every $C\in\SConf{0}_M$,
  \begin{multline*}
    1 \geq \tek{C} W_MW_M^*\ket{C} \\
    = \sum_{C'\in \B_c}\tek{C'} W_MW_M^*\ket{C'}
    - \sum_{C'\in \B_c\setminus\{C\}}\tek{C'} W_MW_M^*\ket{C'} \\
    = |\B_C| - \sum_{C'\in \B_C\setminus\{C\}}\tek{C'} W_MW_M^*\ket{C'} \\
    \geq |\B_C| - |\B_C\setminus\{C\}|
    = 1
  \end{multline*}
  That is, $\tek{C} W_MW_M^*\ket{C} =1$.
\end{proof}

\begin{theorem}\label{thm:qtm-iff-unitary}
  A pre-QTM is a QTM iff its time evolution operator is unitary.
\end{theorem}
\begin{proof}
  If the time evolution operator $U_M$ of the pre-QTM $M$ is unitary,
  and therefore an isometry, then the local unitary conditions hold
  (by Lemma~\ref{lem:isometry-cond}), and thus $M$ is a QTM. On the
  other hand, if $M$ is a QTM, and therefore the unitary conditions
  hold, then $U_M$ is an isometry (by Lemma~\ref{lem:isometry-cond}),
  and it is unitary indeed (by Lemma~\ref{lem:isometry-imp-unitary}).
\end{proof}